\documentclass[authoryear, 3p]{elsarticle}
\makeatletter
\def\ps@pprintTitle{%
 \let\@oddhead\@empty
 \let\@evenhead\@empty
 \def\@oddfoot{}%
 \let\@evenfoot\@oddfoot}
\makeatother
\usepackage{etex}
\usepackage{dsfont}
\usepackage{times}
\usepackage{dsfont}
\usepackage{amsmath}
\usepackage{cite}
\usepackage{forest}
\usepackage{times}
\usepackage{helvet}
\usepackage{courier}
\usepackage[utf8]{inputenc} 
\usepackage{mathtools}
\usepackage{complexity}
\usepackage{balance}
\usepackage{mathdots}
\usepackage{bbold}
\RequirePackage{latexsym}
\RequirePackage{mathtools}
\RequirePackage{amsfonts}
\usepackage{amssymb}
\RequirePackage{amsthm}
\RequirePackage{array}
\RequirePackage{xspace}
\usepackage{url}
\usepackage{graphicx}

\usepackage{newfloat}
\DeclareFloatingEnvironment[
    fileext=los,
    listname=List of Algorithms,
    name=Algorithm,
    placement=tbhp,
    within=none,
]{algorithmenv}
\newcolumntype{d}[1]{D{.}{.}{#1} }

\DeclareFloatingEnvironment[
    fileext=los,
    listname=List of Algorithms,
    name=Procedure,
    placement=tbhp,
    within=none,
]{procedurenv}
\newcolumntype{d}[1]{D{.}{.}{#1} }

\newcommand{\ie}{i.e., }
\newcommand{\eg}{e.g., }

\newcommand{\wrt}{w.r.t.\ }

\let\oldnl\nl
\newcommand{\nonl}{\renewcommand{\nl}{\let\nl\oldnl}}

\usepackage[textsize=small]{todonotes}

\newcommand*\sizedcircled[2]{\tikz[baseline=(char.base)]{ \node[shape=circle,draw,inner sep=2pt, scale=#1] (char) {\textbf{#2}}; }}

\newcommand{\RR}{\mathbf{R}\xspace}

\newcommand{\N}{\mathbf{N}\xspace}
\newcommand{\Z}{\mathbf{Z}\xspace}

\newcommand{\STN}{STN\xspace}
\newcommand{\MPG}{MPG\xspace}

\newcommand{\HTN}{HyTN\xspace}

\newcommand{\HTNC}{HyTN-Consistency\xspace}
\newcommand{\CSTP}{CSTP\xspace}
\newcommand{\CSTN}{CSTN\xspace}
\newcommand{\DCC}{DC\xspace}
\newcommand{\eDCC}{$\epsilon$-DC\xspace}
\newcommand{\oDCC}{$\pi$-DC\xspace}
\newcommand{\Hst}{\emph{Hst}\xspace}
\newcommand{\oHst}{\emph{$\pi$-Hst}\xspace}
\newcommand{\Con}{\emph{Con}\xspace}

\newcommand{\Sub}{\emph{Sub}\xspace}

\def\Ord{{\cal O}}

\def\H{{\cal H}}

\def\A{{\cal A}}

\def\S{{\cal S}}

\newcommand{\figref}[1]{Fig.~\ref{#1}}

\newtheorem{Thm}{Theorem}

\newtheorem{Lem}[]{Lemma}
\newtheorem{Prop}[]{Proposition}

\theoremstyle{plain}
\newtheorem{Def}{Definition}
\newtheorem{Rem}{Remark}
\newtheorem{Ex}{Example}

\usepackage[ruled,vlined,linesnumbered]{algorithm2e}
\usepackage{subfig}
\SetCommentSty{textsf}
\SetKwRepeat{DoWhile}{do}{while}
\SetAlFnt{\small}
\makeatletter
\newcommand{\removelatexerror}{\let\@latex@error\@gobble}
\makeatother

\usepackage{tikz}
\usepackage{tikz-qtree}
\usetikzlibrary{calc,positioning,fit}
\usetikzlibrary{shapes,shapes.multipart,shapes.arrows}
\usetikzlibrary{decorations,decorations.pathmorphing,decorations.pathreplacing,decorations.markings,decorations.shapes}
\usetikzlibrary{arrows}
\usetikzlibrary{fit,backgrounds}
\usetikzlibrary{plotmarks}
\tikzstyle{node}=[circle,draw,inner sep=2pt,transform shape,minimum size=1.75em]
\tikzstyle{task}=[draw,rectangle,inner sep=1.5pt,transform shape
                         ,font=\sffamily\scriptsize,text width=27pt
                         ,text centered,minimum height=2em
                         ]
\tikzstyle{connector}=[draw,diamond,shape aspect=1,inner sep=1pt,transform shape
                 ,font=\sffamily\scriptsize,text width=15pt
                 ,text centered
                 ]

\tikzstyle{StartCase}=[circle,draw,minimum size=.75cm,transform shape]
\tikzstyle{EndCase}=[circle,draw,ultra thick,minimum size=.75cm,transform shape]
\tikzstyle{smallLabel}=[font=\sffamily\normalsize,inner sep=1pt,transform shape]
\tikzstyle{timeLabel}=[smallLabel,midway,transform shape]
\tikzstyle{minWidth}=[text width=.9cm]
\tikzstyle{info}=[rounded corners,fill=yellow,text width=1cm,text centered,inner sep=1pt]
\tikzstyle{infoLine}=[thin,decorate,decoration={snake,amplitude=.4mm, segment length=2mm, post length=1mm}]
\tikzstyle{edgeLabel}=[font=\tiny,sloped]
\tikzstyle{infoRow}=[rounded corners,fill=yellow,inner sep=1pt]
\tikzstyle{punto}=[circle,draw,fill=black,minimum size=2bp,inner sep=0pt,outer sep=0pt]
\tikzstyle{blackNode}=[fill=black!30]
\tikzstyle{crosses}=[decorate,decoration={name=crosses,segment length=2mm,post length=2mm}]
\tikzstyle{every picture}=[>=latex]
\tikzstyle{every label}=[inner sep=2pt]

\begin{document}
\begin{frontmatter}

\title{Instantaneous Reaction-Time in Dynamic Consistency Checking \\
  of Conditional Simple Temporal Networks \\
  \normalsize -- Extended Version with an Improved Upper Bound -- }

\author[tn]{Massimo Cairo}
\ead{massimo.cairo@unitn.it}

\author[vr]{Carlo Comin}
\ead{carlo.comin@univr.it}

\author[vr]{Romeo Rizzi}
\ead{romeo.rizzi@univr.it}

\address[tn]{Department of Mathematics, University of Trento, Trento, Italy}
\address[vr]{Department of Computer Science, University of Verona, Verona, Italy}

\begin{abstract}
Conditional Simple Temporal Networks (\CSTN{s}) is a constraint based graph-formalism for conditional temporal planning.
Three notions of consistency arise for \CSTN{s}: weak, strong, and \emph{dynamic}.
Dynamic-Consistency (DC) is the most interesting notion, but it is also the most challenging.
In order to address the DC-Checking problem, \citet{CR15} introduced $\varepsilon$-DC
(a refined, more realistic, notion of DC), and provided an algorithmic solution to~it.
Next, given that DC implies $\varepsilon$-DC for some sufficiently small $\varepsilon>0$,
and that for every $\varepsilon>0$ it holds that $\varepsilon$-DC implies DC,
it was offered a sharp lower bounding analysis on the critical value of the \emph{reaction-time} $\hat\varepsilon$ under which the two notions coincide.
This delivered the first (pseudo) singly-exponential time algorithm for the DC-Checking of \CSTN{s}.
However, the $\varepsilon$-DC notion is interesting per se, and the $\varepsilon$-DC-Checking algorithm of~\citet{CR15}
rests on the assumption that the reaction-time satisfies $\varepsilon > 0$; leaving unsolved the question of what happens when $\varepsilon = 0$.
In this work, we introduce and study \oDCC, a sound notion of DC with an \emph{instantaneous} reaction-time
(\ie one in which the planner can react to any observation at the same instant of time in which the observation is made).
Firstly, we demonstrate by a counter-example that \oDCC is not equivalent to $0$-DC,
and that $0$-DC is actually inadequate for modelling DC with an instantaneous reaction-time.
This shows that the main results obtained in our previous work do not apply directly, as they were formulated,
to the case of $\varepsilon=0$. Motivated by this observation, as a second contribution,
our previous tools are extended in order to handle \oDCC, and the notion of \emph{ps-tree} is introduced,
also pointing out a relationship between \oDCC and consistency of Hyper Temporal Networks.
Thirdly, a simple \emph{reduction} from \oDCC to (classical) \DCC is identified.
This allows us to design and to analyze the first sound-and-complete \oDCC-Checking algorithm.
Remarkably, the time complexity of the proposed algorithm remains (pseudo) singly-exponential	in the number of propositional letters.
Finally, it is observed that the technique can be leveraged to actually reduce from \oDCC to $1$-\DCC,
this allows us to further improve the exponents in the time complexity of \oDCC-Checking.
\end{abstract}

\end{frontmatter}

\section{Introduction}\label{sect:introduction}
In \emph{temporal planning} and \emph{scheduling}, temporal networks have long been employed for the representation, validation, and execution
of plans affected by temporal constraints \citep{Bettini2002,Chinn2000,CombiGMP14,CombiP09,CombiP10,Eder2000}.
The specification of a temporal network consists of two main components:
\emph{time-points} and \emph{temporal constraints}. Time-points are variables whose domain is the set of real numbers; temporal
constraints are linear inequalities that specify lower and upper bounds on the temporal distance between pairs of time-points
\citep{DechterMP91}. The \emph{execution} of a time-point (\ie the assignment of a real value to it) models the (instantaneous)
occurrence of an event. An agent executing a temporal network aims to execute its time-points so that all relevant temporal constraints are satisfied.

\emph{Simple Temporal Networks} (STNs) are the most studied and used kind of temporal networks due to simplicity, efficiency, and general applicability~\citep{DechterMP91}.
An \STN can be seen as a directed weighted graph, where nodes represent events to be scheduled in time (time-points) and
arcs represent temporal distance (difference) constraints between pairs of events.
An STN is typically used in planning applications where all time-points must be executed (\ie
must play their role in the plan) and where the agent controls their execution.
An STN is {\em consistent} if the network can be executed in such a way as to satisfy all of its temporal constraints.

Since STNs were proposed, several authors have introduced extensions to STNs to augment their expressiveness.
This work is focused on the \emph{Conditional Simple Temporal Problem (\CSTP)}~\citep{TVP2003}
and its graph-based counterpart \emph{Conditional Simple Temporal Network (\CSTN)}~\citep{HPC12},
a constraint-based model for conditional temporal planning. The \CSTN formalism extends {\STN}s in that:
(1) some of the nodes are called \emph{observation time-points} and to each of them is associated a propositional letter, to be disclosed only at execution time;
(2) \emph{labels} (\ie conjunctions over the literals) are attached to all nodes \emph{and} constraints,
to indicate the scenarios in which each of them is required. The planning agent (planner) must schedule all the required nodes,
while respecting all the required temporal constraints among them.
This extended framework allows for the off-line construction of conditional plans that are guaranteed to satisfy complex
temporal constraints. Importantly, this can be achieved even while allowing the decisions about the precise timing of actions
to be postponed until execution time, in a least-commitment manner, thereby adding flexibility and making it possible to
adapt the plan dynamically, during execution, in response to the observations that are made~\citep{TVP2003}.

Three notions of consistency arise for \CSTN{s}: weak, strong, and \emph{dynamic}.
Dynamic-Consistency (\DCC) is the most interesting one, as it requires the existence of conditional plans where
decisions about the precise timing of actions are postponed until execution time, but anyhow guaranteeing that
  all of the relevant constraints will be ultimately satisfied.
Still, \DCC is the most challenging one and it was conjectured to be hard to assess~\citep{TVP2003}.

In~\citet{CR15}, it was unveiled that \emph{Hyper Temporal Networks (\HTN{s})} and \emph{Mean-Payoff Games (\MPG{s})}
  are a natural underlying combinatorial model for the DC-Checking of \CSTN{s}.
Indeed, \STN{s} have been recently generalized into \emph{Hyper Temporal Networks (\HTN{s})}~\citep{CPR2014, CPR2016}, by considering
weighted directed hypergraphs, where each hyperarc models a disjunctive temporal constraint called a \emph{hyperconstraint}.
In~\citet{CPR2014, CPR2016}, the computational equivalence between checking the consistency of \HTN{s} and determining
the winning regions in \MPG{s} was also pointed out; the approach was shown to be viable and robust thanks to some extensive experimental evaluations.
Recall that \MPG{s}~\citep{EhrenfeuchtMycielski:1979, ZwickPaterson:1996, brim2011faster, CominR17}
are a family of two-player infinite games played on finite graphs,
with direct and important applications in model-checking and formal verification of finite-state reactive systems~\citep{Gradel2002};
also, they are known for having theoretical interest in computational complexity
because of their special place among the few (natural) problems lying in $\NP\cap\coNP$.

In~\citet{CR15}, the first (pseudo) singly-exponential time algorithm for solving the DC-Checking problem appeared,
also producing a dynamic execution strategy whenever the input \CSTN is DC.
For this, they introduced $\varepsilon$-DC (a refined, more realistic, notion of DC), and provided the first algorithmic solution to it.
Next, given that DC implies $\varepsilon$-DC for some sufficiently small $\varepsilon>0$,
and that for every $\varepsilon>0$ it holds that $\varepsilon$-DC implies DC, it was offered a sharp lower bounding analysis on
the critical value of the \emph{reaction-time} $\hat\varepsilon$ under which the two notions coincide.
This delivered the first (pseudo) singly-exponential time algorithm for the DC-Checking of \CSTN.
However, the $\varepsilon$-DC notion is interesting per se, and the $\varepsilon$-DC-Checking algorithm in~\citet{CR15}
rests on the assumption that the reaction-time satisfies $\varepsilon > 0$; leaving unsolved the question of what happens when $\varepsilon = 0$.

\subsection{Contribution}
In this work we introduce and study \oDCC, a sound notion of DC with an \emph{instantaneous} reaction-time
(one in which the planner can react to any observation at the same instant of time in which the observation is made).

Our contributions concerning \oDCC are summarized below.
\begin{enumerate}
\item We provide a neat counter-example showing that \oDCC is not just the $\varepsilon=0$ special case  of $\varepsilon$-DC.
Moreover, that the semantics of the $\varepsilon=0$ special case  of $\varepsilon$-DC is conceptually flawed;
  thus, to the best of our knowledge,
    \oDCC provides the first sound formalization of the semantics of instantaneous reaction-time DC.

This also implies that the algorithmic results obtained in~\citet{CR15} do not apply directly
to the study of those situations where the planner is allowed to react instantaneously,
  so the algorithmics needs to be extended and adapted to the proposed $\pi$-DC semantics in a non-trivial way.
\item Motivated by this observation, as a second contribution, we extend the previous formulation \citep{CR15}
to capture a sound notion of DC with an instantaneous reaction-time.
It turns out that \oDCC needs to consider an additional inner ordering among all the observation nodes that occur at the same instant of time.
\item The notion of \emph{permutation-scenario tree (ps-tree)} is introduced to reflect the ordered structure of a \oDCC execution strategy,
also allowing us to point out a relationship between \oDCC and consistency of \HTN{s}.
\item Thirdly, a simple \emph{reduction} from \oDCC to (classical) \DCC is identified.
This allows us to design and to analyze the first sound-and-complete \oDCC-Checking procedure.
Remarkably the time complexity of the proposed algorithm remains (pseudo) singly-exponential in the number $|P|$ of propositional letters.
\item It is further observed that the same technique can be leveraged to actually reduce from
  \oDCC to $\epsilon$-\DCC with $\epsilon=1$, \ie to reduce \oDCC to $1$-\DCC;
this allows us to further improve the exponents in
the time complexity of the proposed \oDCC-Checking algorithm (see Theorem~\ref{thm:mainresult_pi-DC} for the actual time bounds).
\end{enumerate}

\paragraph{Delta-Contribution}
Notice that a preliminary version of this article (\citet{CCR16}) appeared in the proceedings of the
\emph{23rd International Symposium on Temporal Representation and Reasoning (TIME 2016)}.

Here the previous results are improved and the presentation is extended in many aspects, including:
\begin{enumerate}
\item the time complexity of \oDCC checking (Algorithm~\ref{algo:check_pi-DC}) is further improved by a factor
of $|\Sigma_P|\cdot |V|$ \wrt \citet{CCR16} (see Theorem~\ref{thm:mainresult_pi-DC}),
where $|\Sigma_P|$ is the number of possible execution scenarios (\ie $|\Sigma_P|\leq 2^{|P|}$,
  notice that this is an exponential factor)
and $|V|$ is the total number of time-points in the input \CSTN; the improvement is achieved by reducing \oDCC to $1$-\DCC,
and then by solving $1$-DC with the $\epsilon$-DC checking algorithm devised in~\citet{CR15}.
\item it is offered an updated account of the current literature that is related to the contributions offered in this work,
  this is done in the next subsection which is devoted to the \emph{Motivations and Related Works};
\item new figures and examples are added to better clarify the technical constructions and key arguments;
\item new remarks are added to point out technical key facts,
  \eg that we choose to focus on integer weighted networks for ease of notation,
  meanwhile supporting rational weights as a straightforward generalization.
\item new paragraphs are devoted to provide the underlying intuitions so that
to better guide the reader among the examples, technical definitions, lemmata and theorems.
\end{enumerate}
\subsection{Motivations and Related Works}\label{sect:relatedworks}
This section discusses of some related approaches offered in the current literature.
CSTNs have been implicitly proposed for the first time in~\citet{TVP2003},
where authors formally introduce the Conditional Simple Temporal Problem (called CTP instead of CSTP). CSTP consists of determining whether a given CSTN admits a viable dynamic execution strategy. In the paper there is no a formal definition of CSTN and propositional labels are associated only to nodes, while they are implicit for constraints\slash edges.
Moreover, authors showed how to solve a CTP, by encoding it as a meta-level Disjunctive Temporal Network (DTN) and feeding it to an off-the-shelf DTN solver. Although of theoretical interest, this approach is not practical because the CSTP-to-DTN encoding has exponential size and, on top of that, the DTN solver runs in exponential time. To our knowledge, this approach has never been implemented or empirically evaluated.
Finally, authors discussed some supplementary reasonable assumptions that any well-defined CSTP must satisfy without formalizing them.

Later on, those conditions have been analyzed and formalized in \citet{HPC12} and in \citet{HPC15}, leading to the sound notion of Conditional Simple Temporal Network (CSTN).
In more details, in \citet{HPC15} both time-point (nodes) and constraints (edges) of a CSTN can have propositional labels---for specifying in which scenarios they have to be considered---and such labels have to satisfy some well-definedness properties, in order to guarantee that a dynamic execution strategy can exists.
Finally, in \citet{HPC15} authors proposed a sound-and-complete algorithm for solving CSTP showing---by an experimental evaluation---its good performance.

\citet{Ci14} provided the first sound-and-complete procedure for checking
the Dynamic-Controllability of \CSTN{s} with Uncertainty (CSTNUs) and this algorithm
can be employed for checking \DCC on \CSTN{s} as a special case. Their approach is based on reducing
the problem to solving Timed Game Automata (TGA). However, solving TGAs is a problem of much higher complexity than solving \MPG{s}.
Indeed, no upper bound is given in~\citet{Ci14} on the time complexity of their solution.
Moreover, neither $\varepsilon$-DC nor any other notion of DC with an instantaneous reaction-time are dealt with in that work.

\citet{Cairo17} introduced a streamlined version of a CSTN in which propositional labels may appear on constraints, but not on time-points. This change simplifies the definition of the DC property, as well as the propagation rules for the DC-checking algorithm.
It also simplifies the proofs of the soundness and completeness of those rules.
This paper provided two translations from traditional CSTNs to streamlined CSTNs.
Each translation preserves the DC property and, for any DC network,
ensures that any dynamic execution strategy for that network can be extended to a strategy for its streamlined counterpart. Finally, the paper presented an empirical comparison of two versions of the DC-checking algorithm:
the original version and a simplified version for streamlined CSTNs.
The comparison was based on CSTN benchmarks from earlier work.
So \citet{Cairo17} observed that the simplified algorithm is a practical alternative
for checking the dynamic consistency of CSTNs.
Unfortunately, no notion of DC with an instantaneous reaction-time is studied in that work.

To the best of our knowledge, the first work to approach a notion of DC
with an instantaneous reaction-time is~\citet{HPC15}, where the corresponding notion was named IR-DC (DC with instantaneous reaction).
The aim was to offer a sound-and-complete propagation-based DC-checking algorithm for CSTNs.
The subsequent work~\citet{HP16} extended and amended~\citet{HPC15} so that to check $\varepsilon$-DC,
both for $\varepsilon>0$ and for $\varepsilon=0$. However, to the best of our knowledge,
  the worst-case complexity of those algorithms is currently unsettled.
  Most importantly, a preliminary conference version (\citet{CCR16}) of the present work inspired \citet{HunsIJCAI18} to
    show that the semantics of IR-DC is also flawed
    due to the same fact pointed out in this work, \ie
      that an \emph{instantaneous circularity} emerges for carefully constructed \CSTN{s}.

Indeed, soon after the appearance of a preliminary conference version (\citet{CCR16}) of this article,
our proposed $\pi$-DC notion inspired the realization of a sequence of relevant quality papers.

Most notably,~\citet{HunsIJCAI18} showed that:
(1) the IR-DC semantics is also flawed, thanks to a similar argument;
(2) that one of the constraint-propagation rules from the IR-DC-checking algorithm is not sound with respect to the IR-DC semantics;
(3) as an aftermath, it was presented a simpler constraint-propagation algorithm, called the $\pi$-DC-checking algorithm;
(4) it was proved that it is sound and complete with respect to our $\pi$-DC semantics;
and (5) the algorithm was empirically evaluated, remarkably the authors found that their $3$-rule
 $\pi$-DC Checking algorithm is much faster than their $6$-rule IR-DC Checking algorithm,
  moreover, they observed that the performance improvement increases as the instance size increases.
Unfortunately no upper bound is given in~\citet{HunsIJCAI18} on the time complexity of their $\pi$-DC checking solution.

Last but not least,~\citet{HunsTIME18} proves that the $\epsilon$-DC checking problem for CSTNs can
be reduced to the standard DC-checking problem for CSTNs -- without incurring any computational cost.
Two versions of these results are presented that differ only in whether a dynamic strategy
can react instantaneously to observations, or only after some arbitrarily small, positive delay.
Remarkably, it was shown that the $\epsilon$-DC checking problem is reducible to the $\pi$-DC checking problem,
 and it was offered experimental evidence that (as soon as the size of the \CSTN{s} instances is sufficiently large)
 adopting $\pi$-DC is a competitive practical approach for solving DC-checking problems.
Again, no upper bound is given in~\citet{HunsTIME18} on the time complexity of their $\pi$-DC checking solution.

All in, apart from their solid theoretical relevance, these recent results discussed above provide
a clear evidence that the $\pi$-DC notion is not a purely theoretical speculation but it can
  actually lead to competitive (and, as the size of the network increases, even more efficient)
  algorithms \wrt those already known in the realm of \CSTN{s} DC-checking.

To the best of our knowledge,
  the improved upper bound on the time complexity of $\pi$-DC offered in the present article
(established in Theorem~\ref{thm:mainresult_pi-DC})
  is currently the state of the art on that matter,
    it seems worthwhile to remark it.

In summary, we believe that the present work can possibly help in clarifying DC
with an instantaneous reaction-time, also providing the first formally proved time complexity (improved)
  upper bounds that can also act as a reference point for other recent works in the literature;
and we believe that this might turn out to be helpful when the perspective has to be that of providing
sound-and-complete algorithms based on the propagation of LTCs.

\section{Background}\label{sect:backgroundandnotation}
This section provides some background and preliminary notations.
To begin, if $G=(V,A)$ is a directed weighted graph,
every arc $a\in A$ is a triplet $(u,v,w_a)$ where $u=t(a) \in V$ is the \textit{tail} of $a$,
$v=h(a) \in V$ is the \textit{head} of $a$, and $w_a=w(u,v)\in\Z$ is the (integer) \textit{weight} of $a$.
Let us now recall the definition of Simple Temporal Networks (\STN{s}).
\begin{Def}[\STN{s}~\citep{DechterMP91}]
An \STN is a weighted directed graph whose nodes are time-points that must be placed on the real time line and whose arcs
express binary distance constraints on the allocations of their end-points in time.
An \STN $G=(V, A)$ is \textit{consistent} if it admits a \emph{feasible schedule},
\ie a schedule $\phi: V\mapsto \RR$ such that $\phi(v) \leq \phi(u) + w(u,v)$ for all arcs $(u,v, w(u,v))\in A$.
\end{Def}

\subsection{Conditional Simple Temporal Networks}\label{subsect:CSTN}
Let us briefly recall the \CSTN model from~\citet{TVP2003, HPC12, CR15},
where we refer the reader who would like to find more extensive exposures along with illustrative examples.

We'll need to introduce a little bit of notation at first.
Let $P$ be a set of propositional letters (\ie boolean variables), a \emph{label} is any (possibly empty) conjunction of letters,
or negations of letters, drawn from $P$. The \emph{empty label} is denoted by $\lambda$.
The set of all these labels is denoted by $P^*$. Two labels, $\ell_1$ and $\ell_2$,
are called \emph{consistent}, denoted by $\Con(\ell_1, \ell_2)$,
when $\ell_1\wedge\ell_2$ is satisfiable. A label $\ell_1$ \emph{subsumes} a label $\ell_2$,
denoted by $\Sub(\ell_1, \ell_2)$, whenever $\ell_1\Rightarrow\ell_2$ holds.

We are now in the position to recall the formal definition of {\CSTN}.
\begin{Def}[{\CSTN}s~\citet{TVP2003, HPC12}]
A \emph{Conditional Simple Temporal Network (\CSTN)} is a tuple $( V, A, L, \Ord, \Ord{V}, P )$ where:
\begin{itemize}
\item $V$ is a finite set of \emph{time-points}; $P=\{p_1, \ldots, p_{|P|}\}$ is a finite set
	of \emph{propositional letters};
\item $A$ is a set of \emph{labelled temporal constraints (LTCs)} each having the form $( v-u\leq w(u,v), \ell)$,
	where $u,v\in V$, $w(u,v) \in \Z$, and the label $\ell\in P^*$ is satisfiable;
\item $L:V\rightarrow P^*$ assigns a label to each time-point in $V$;
\item $\Ord{V}\subseteq V$ is a finite set of \emph{observation time-points};
	$\Ord:P\rightarrow \Ord{V}$ is a bijection that associates a unique
	observation time-point $\Ord(p)=\Ord_p$ to each proposition $p\in P$;
\item The following \emph{reasonability assumptions} must hold:

(\emph{WD1}) for any LTC $( v-u\leq w, \ell)\in A$ the label $\ell$ is satisfiable
and subsumes both $L(u)$ and $L(v)$; intuitively, whenever a constraint $( v-u\leq w)$ is required,
then its endpoints $u$ and $v$ must be scheduled (sooner or later);

(\emph{WD2}) for each $p\in P$ and each $u\in V$ such that either $p$ or $\neg p$ appears in $L(u)$, we require:
$\Sub(L(u), L(\Ord_p))$, and $( \Ord_p-u\leq -\epsilon, L(u)) \in A$ for some small $\epsilon>0$;
intuitively, whenever a label $L(u)$, for some $u\in V$,
contains some $p\in P$, and $u$ is eventually scheduled,
then $\Ord_p=\Ord(p)$ must be scheduled before the time of $u$.

(\emph{WD3}) for any LTC $( v-u\leq w, \ell)\in A$ and $p\in P$,
for which either $p$ or $\neg p$ appears in $\ell$, it holds that $\Sub(\ell, L(\Ord_p))$;
intuitively, assuming that a required constraint contains some $p\in P$,
then $\Ord_p=\Ord(p)$ must be scheduled (sooner or later).
\end{itemize}
\end{Def}

We shall adopt the notation $x\overset{[a,b], \ell}{\longrightarrow} y$,
where $x,y\in V$, $a,b\in \N, a<b$ and $\ell\in P^*$,
to compactly represent the LTCs $( y-x\leq b, \ell ), ( x-y\leq -a, \ell )\in A$;
also, whenever $\ell=\lambda$, we shall omit $\ell$ from the graphics,
 see~\eg~\figref{fig:example1}.

\begin{Rem} Once the $\pi$-DC notion will be formally defined in the forthcoming sections,
 we shall relax the \emph{WD2} assumption
 by allowing $\Ord_p$ to	be scheduled at the same instant of time of $u\in V$ (but, in case $u\in\Ord V$,
 still at a subsequent \emph{position} in the additional inner ordering
 among the observation time-points that the $\pi$-DC introduces).
\end{Rem}

\begin{Rem}
Since the algorithmic framework that we will take as a reference exhibits
	a pseudo-polynomial running time complexity,
it will make sense for us to focus on temporal networks having integer weights only
(\ie each arc of the network will be weighted with an integer number).
Please notice that this already handles the case in which the temporal-constraints
	are weighted	with rational numbers instead
(provided they are appropriately scaled back so that to be treated as if they were integers).
Indeed, in the rational case, each arc weight can be expressed
as a fraction involving the least common denominator among all of the arc weights. As a result,
without loss of generality, we may henceforth assume that all arc weights are integers.
Also notice that from the perspective of the currently known algorithmic framework,
much less concrete sense has for us to consider real numbers with a possibly infinite decimal expansion;
given the numerical-arithmetic time complexity of basically all of the currently known procedures,
in any case we would find ourselves having to truncate the real decimal expansion at some point
(in order to land on a decision procedure that finally halts for real),
		ultimately considering rational numbers to all effects.

In summary, all of our definitions and results will be formulated in the realms on the integers, for ease of notation,
but they all extend in a very natural and straightforward way to the rational numbers too (a scaling-factor appears).
\end{Rem}

\begin{Ex}\label{example1}
Fig.~\ref{fig:example1} depicts an example \CSTN $\Gamma_0$
	having three (non observation) time-points $A$, $B$ and $C$
		as well as two observation time-points $\Ord_p$ and $\Ord_q$.

The underlying intuition being that: time-point $A$ will serve as a \emph{zero node}
 (\ie one which is required to be scheduled before all the others), and $C$ will serve as a \emph{tail-light}
 (\ie one which is required to be scheduled lastly,
	in this case exactly $10$ units of time after the zero $A$);
	then, time-point $B$ will take the role of a \emph{weathervane} (\ie it will be required
	to be scheduled either soon after the zero $A$ or shortly before the tail-light $C$,
	depending on the outcome of the two observation time-points $\Ord_p$ and $\Ord_q$).
	If $p$ turns out to be \emph{true}, then the propositional value of $q$ will be decisive for scheduling $B$;
	otherwise, it will be convenient just to schedule $B$ very close to $C$.

	Therefore,
	it will make sense to schedule $\Ord_p$ to happen \emph{strictly} before $\Ord_q$ (see Example~\ref{ex:strategyexample}).

Formally, the \CSTN $\Gamma_0=( V, A, L, \Ord, \Ord{V}, P)$ is defined as follows: $V=\{A,B,C,\Ord_p, \Ord_q\}$, $P=\{p,q\}$, $\Ord{V}=\{\Ord_p, \Ord_q\}$,
$L(v)=\lambda$ for every $v\in V\setminus\{\Ord_q\}$ and $L(\Ord_q)=p$, $\Ord(p)=\Ord_p, \Ord(q)=\Ord_q$.
Next, the set of LTCs is: $A=\{ ( C-A\leq 10, \lambda ), ( A-C\leq -10,
\lambda ), ( B-A\leq 3, p\wedge\neg q ),
( A-B\leq 0, \lambda ),
( \Ord_p-A\leq 5, \lambda ),
( A-\Ord_p\leq 0, \lambda ),
( \Ord_q-A\leq 9, p ),
( A-\Ord_q\leq 0, p ),
( C-B\leq 2, q ),
( C-\Ord_p\leq 10, \lambda)$.
\end{Ex}

Sometimes we will show the scheduling time of a time-point with a label in boldface
 on the sidelines of the node itself,
 	see \eg $A$ in \figref{fig:example1} (\ie the one which is scheduled at time $t=0$).

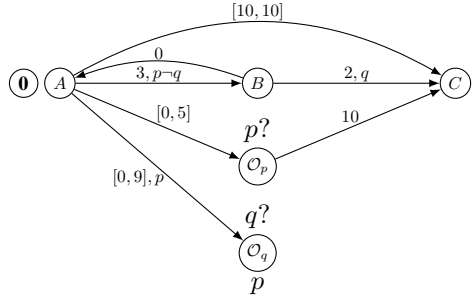
\begin{figure}[!h]
\centering
\begin{tikzpicture}[arrows=->,scale=.65,node distance=1 and 1]
 	\node[node, label={left:$\bf \sizedcircled{.75}{0}$}] (A') {$A$};
	\node[node, xshift=15ex,right=of A'] (B') {$B$};
	\node[node, xshift=15ex,right=of B'] (C') {$C$};
	\node[node,below=of B', label={above,yshift=0ex:$p?$}] (P') {$\Ord_p$};
	\node[node,below=of P', label={above,yshift=0ex:$q?$}, label={below:$p$}] (Q') {$\Ord_q$};
	\draw[] (A') to [bend left=30] node[timeLabel,above] {$[10,10]$} (C');
	\draw[] (A') to [] node[timeLabel,above] {$3, p \neg q$} (B');
  \draw[] (B') to [bend right=20] node[timeLabel,above] {$0$} (A');
	\draw[] (B') to [] node[timeLabel,above] {$2, q$} (C');
	\draw[] (A') to [] node[xshift=2.1ex, yshift=0ex, timeLabel,above] {$[0,5]$} (P');
	\draw[] (A') to [] node[xshift=-2.5ex,yshift=0ex, timeLabel,below] {$[0,9],p$} (Q');
	\draw[] (P') to [] node[xshift=-1ex, timeLabel,above] {$10$} (C');
	\end{tikzpicture}
\caption{An example CSTN.}\label{fig:example1}
\end{figure}

In all of the following definitions, which basically come from~\citep{TVP2003, HPC15},
we shall implicitly refer to some \CSTN $\Gamma=( V, A, L, \Ord, \Ord{V}, P )$.

Intuitively, a \emph{scenario} denotes a state of affairs, a configuration of the observable environment,
where the set of events that have been observed can be either partial or complete (\ie total).

\begin{Def}[Scenario]
A \emph{scenario} over a subset $U\subseteq P$ of boolean	variables is a truth assignment $s:U\rightarrow \{0, 1\}$,
\ie $s$ is a function that assigns a truth value to each proposition $p\in U$.
When $U\subsetneq P$ and $s:U\rightarrow \{0, 1\}$, then $s$ is said to be a \emph{partial} scenario;
otherwise, when $U=P$, then $s$ is said to be a \emph{(complete)} scenario.
The set comprising all of the complete scenarios over $P$ is denoted by $\Sigma_P$.
If $s\in\Sigma_P$ is a scenario and $\ell\in P^*$ is a label, then $s(\ell)\in\{0,1\}$
denotes the truth value of $\ell$ induced by $s$ in the natural way.
\end{Def}
Notice that any scenario $s\in\Sigma_P$ can be described by means of the label
$\ell_s\triangleq l_1\wedge\cdots\wedge l_{|P|}$ such that, for every $1\leq i\leq |P|$,
the literal $l_i\in\{p_i, \neg p_i\}$ satisfies $s(l_i)=1$.
\begin{Ex}
Consider the set of boolean variables $P=\{p,q\}$.
The scenario $s:P\rightarrow\{0, 1\}$ defined as $s(p)=1$
and $s(q)=0$ can be compactly described by the label $\ell_s=p\wedge \neg q$.
\end{Ex}

\begin{Def}[Schedule] A \emph{schedule} for a subset of time-points $U\subseteq V$ is a function $\phi:U\rightarrow\RR$
that assigns a real number to each time-point in $U$. The set of all schedules over $U$ is denoted by~$\Phi_U$.
\end{Def}
\begin{Def}[Scenario-Restriction]
Let $s\in\Sigma_{P}$ be a scenario.
The \emph{restriction} of $V$, $\Ord{V}$, and $A$ \wrt $s$ is defined as:
\begin{itemize}
\item $V^+_s\triangleq \big\{v\in V\mid s(L(v))=1\big\}$; $\Ord{V}^+_s\triangleq \Ord{V}\cap V^+_s$;
\item $A^+_s\triangleq \big\{( u,v,w) \mid \exists {\ell}\, ( v-u\leq w, \ell) \in A, s(\ell)=1\big\}$.
\end{itemize}
The restriction of $\Gamma$ \wrt $s\in \Sigma_P$ is the \STN $\Gamma^+_s\triangleq ( V^+_s, A^+_s)$.
Finally, it is worth to denote $V^+_{s_1, s_2} \triangleq V^+_{s_1}\cap V^+_{s_2}$.
\end{Def}

\begin{Ex}
\figref{FIG:restriction_cstn2} depicts the restriction \STN ${\Gamma_0}^+_s$ of the \CSTN $\Gamma_0$
(Example~\ref{example1}), \wrt the scenario $s(p)=s(q)=0$.
\end{Ex}
\begin{figure}[!htb]
\centering
\begin{tikzpicture}[arrows=->,scale=.65,node distance=1 and 1]
 	\node[node, label={left:$\bf \sizedcircled{.65}{0}$}] (A') {$A$};
	\node[node, xshift=15ex,right=of A', label={right:$\bf \sizedcircled{.65}{8}$}] (B') {$B$};
	\node[node, xshift=15ex,right=of B', label={right:$\bf \sizedcircled{.65}{10}$}] (C') {$C$};
	\node[node,below=of B', label={below:$\bf \sizedcircled{.65}{1}$}] (P') {$\Ord_p$};
	\draw[] (A') to [bend left=30] node[timeLabel,above] {$[10,10]$} (C');
  \draw[] (B') to [] node[timeLabel,above] {$0$} (A');
	\draw[] (A') to [] node[xshift=2ex, yshift=0ex, timeLabel,above] {$[0,5]$} (P');
	\draw[] (P') to [] node[xshift=-1ex, timeLabel,above] {$10$} (C');
\end{tikzpicture}
\caption{The restriction ${\Gamma_0}^+_s$ (a) \wrt the
	scenario $s(p)=s(q)=0$.}\label{FIG:restriction_cstn2}
\end{figure}
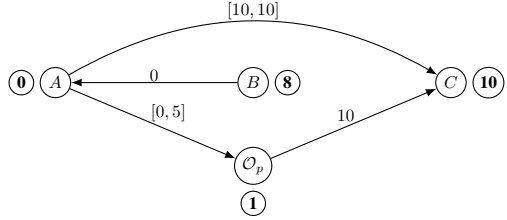

Intuitively, a scenario-restriction projects the \CSTN  by confirming only certain temporal-constraints as active meanwhile
 deleting all those that are not prescribed by the referenced scenario;
 notice that, in the projected network, the temporal constraints will be unlabeled so that the projected \CSTN will become an \STN actually.

Instead, an execution-strategy for a \CSTN is actually a family of schedule assignments
	(one for each complete scenario $s\in\Sigma_P$),
		in which only the time-points confirmed by the corresponding scenario will be scheduled over time.

\begin{Def}[Execution-Strategy]\label{def:executionstrategy}
An \emph{Execution-Strategy (ES)} for $\Gamma$ is a mapping
$\sigma:\Sigma_P\rightarrow \Phi_{V}$ such that,
for any complete scenario $s\in\Sigma_P$, the domain of the schedule $\sigma(s)$ is $V^+_{s}$.
The set of ESs of $\Gamma$ is denoted by $\S_{\Gamma}$.

The \emph{execution time} of a time-point $v\in V^+_{s}$ in the
schedule $\sigma(s)\in\Phi_{V^+_s}$ is denoted~by~$[\sigma(s)]_v$.
\end{Def}

In order to sustain a formal definition of dynamic-consistency, it is worth considering
a notion of \emph{history} depending on a particular ES $\sigma$, a scenario $s$,
and a real-value $\tau$ as in from~\citet{HPC12};
whose intended interpretation is the set comprising all	and only those
propositional letters $p$ (together with their propositional value $s(p)$) such that
	their observation time-point $\Ord_p$ is required under scenario $s$ and
	it is scheduled by $\sigma(s)$ \emph{strictly} before time~$\tau$.

\begin{Def}[History]\label{def:scenario_history}
Let $\sigma\in\S_{\Gamma}$ be any ES, let $s\in\Sigma_P$ be any scenario and let $\tau\in\RR$.
The \emph{history} $\Hst(\tau,s,\sigma)$ of $\tau$ in the scenario $s$ under strategy $\sigma$ is defined as:
$\Hst(\tau,s,\sigma)\triangleq \big\{\big(p, s(p)\big)\in P\times\{0,1\}\mid
	\Ord_p \in V^+_{s},\, [\sigma(s)]_{\Ord_p} < \tau \big\}$.
\end{Def}

Notice that the history can be compactly encoded as the conjunction of the literals
corresponding to the observations comprising it, that is, by means of a label.
\begin{Def}[Viable Execution-Strategy]
We say that $\sigma\in\S_{\Gamma}$ is a \emph{viable} ES if, for each scenario $s\in\Sigma_P$,
the schedule $\sigma(s)\in\Phi_{V^+_{s}}$ is feasible for the \STN $\Gamma^+_s$.
\end{Def}
We are now in the position to recall the formal definition of Dynamic-Consistency from~\citet{TVP2003, HPC12}.
Intuitively, it requires the existence of conditional plans where
decisions about the precise timing of actions are postponed until execution time, but anyhow guaranteeing that
  all of the relevant constraints will be ultimately satisfied;
	stated otherwise, the planning decisions are
		allowed to depend on past observations only
			(\ie the planner can observe the past and react to it), but not on future events.

\begin{Def}[Dynamic-Consistency]\label{def:consistency}
An ES $\sigma\in \S_{\Gamma}$ is called \emph{dynamic} if, for any $s_1, s_2\in \Sigma_P$ and any $v\in V^+_{s_1,s_2}$,
the following implication holds on the scheduling time of $v$ under $s_1$, say $\tau\triangleq [\sigma(s_1)]_v$:
	\[\Con(\Hst(\tau, s_1, \sigma), s_2) \Rightarrow [\sigma(s_2)]_v=\tau.\]
We say that $\Gamma$ is \emph{dynamically-consistent (DC)}
	if it admits $\sigma\in\S_{\Gamma}$ which is both viable and dynamic.

The problem of checking whether a given \CSTN is DC is named \emph{\DCC}-Checking.
\end{Def}

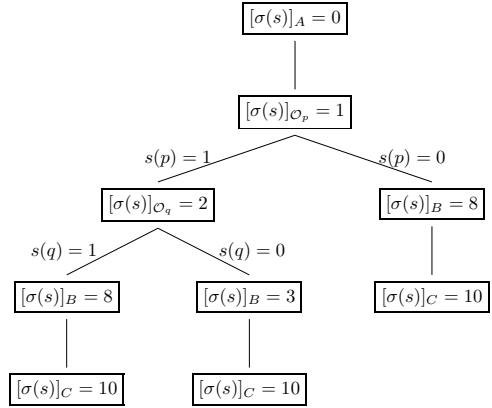
\begin{figure}[!htb]
\centering
\begin{tikzpicture}[scale=0.7, level distance=50pt,sibling distance=30pt]
\Tree [. \framebox{$[\sigma(s)]_{A}=0$}
\edge node[]{};
[. \framebox{$[\sigma(s)]_{\Ord_p}=1$}
\edge node[left,xshift=-1ex]{$s(p)=1$};
[. \framebox{$[\sigma(s)]_{\Ord_q}=2$}
\edge node[left,xshift=-1ex]{$s(q)=1$};
[. \framebox{$[\sigma(s)]_B=8$}
\edge node[]{}; \framebox{$[\sigma(s)]_C=10$}
]
\edge node[right,xshift=1ex]{$s(q)=0$};
[. \framebox{$[\sigma(s)]_B=3$}
\edge node[]{}; \framebox{$[\sigma(s)]_C=10$} ]
]
\edge node[right,xshift=1ex]{$s(p)=0$};
[. \framebox{$[\sigma(s)]_B=8$}
\edge node[]{}; [.
 \framebox{$[\sigma(s)]_C=10$} ]
 ]
]]
\end{tikzpicture}
\caption{A tree-like representation of a dynamic execution strategy $\sigma$ for the CSTN $\Gamma_0$ of \figref{fig:example1},
where $s$ is for the scenario and $[\sigma(s)]_{\cdot}$ is the corresponding scheduling value.}
\label{FIG:cstn2-strategy}
\end{figure}

\begin{Ex}\label{ex:strategyexample}
Consider the CSTN $\Gamma_0$ of \figref{fig:example1},
and let the scenarios $s_1, s_2, s_3, s_4$ be defined as:
$s_1=p\wedge q$; $s_2=p\wedge \neg q$;
$s_3=\neg p\wedge q$; $s_4= \neg p \wedge \neg q$.
The following defines an execution strategy $\sigma\in\S_\Gamma$:
$[\sigma(s_i)]_A=0$ for every $i\in\{1,2,3,4\}$;
$[\sigma(s_i)]_B=8$ for every $i\in\{1,3,4\}$ and $[\sigma(s_2)]_B=3$;
$[\sigma(s_i)]_C=10$ for every $i\in\{1,2,3,4\}$;
$[\sigma(s_i)]_{\Ord_p}=1$ for every $i\in\{1,2,3,4\}$.
The reader can check that $\sigma$ is viable and dynamic. Indeed,
  $\sigma$ admits the tree-like
		representation depicted in Fig.~\ref{FIG:cstn2-strategy} where $\Ord_p$ is scheduled before $\Ord_q$.
\end{Ex}

We provide next the definition of \emph{difference set} $\Delta(s_1; s_2)$.
Intuitively, given two scenarios $s_1$ and $s_2$,
 this is the set comprising all and only those observation time-points $\Ord_p$ that
 are required by the first scenario $s_1$ (not necessarily by the second $s_2$)
 and such that the two scenarios disagree on the corresponding propositional letter $p$,
 	\ie $s_1(p)\neq s_2(p)$.

\begin{Def}[Difference-Set]
Let $s_1, s_2\in\Sigma_P$ be any two scenarios.
The set of observation time-points in $\Ord{V}^+_{s_1}$ at which $s_1$ and $s_2$ differ is denoted by $\Delta(s_1;s_2)$.
Formally, \[ \Delta(s_1; s_2) \triangleq \big\{\Ord_p \in \Ord{V}^+_{s_1} \mid s_1(p)\neq s_2(p)\big\}. \]
\end{Def}
Notice that the difference-set is not commutative,
	\ie generally $\Delta(s_1; s_2)\neq \Delta(s_2;s_1)$.
Also notice that the various definitions of history and dynamic consistency that are used by different
authors~\citep{TVP2003, HPC15, CR15} turn out to be equivalent.

\subsection{Hyper Temporal Networks}\label{subsect:HTN}

This subsection surveys the \textit{Hyper Temporal Network} (\HTN) model,
which is a strict generalization of \STN{s} introduced to overcome the limitation of
considering only conjunctions of constraints but maintaining a practical efficiency in the consistency check of the instances.
In a \HTN a single temporal hyperarc constraint may be defined as a set of two or
more maximum delay constraints which is satisfied when at least one of these delay constraints is satisfied.
\HTN{s} are meant as a light generalization of \STN{s} offering an interesting compromise.
In fact, there exist practical pseudo-polynomial time algorithms for checking consistency
and computing feasible schedules for \HTN{s}, see~\citet{CPR2014, CPR2016}.
To the best of our knowledge this is also the only algorithmic
framework that has been experimentally tested and validated on \HTN{s} (with promising results).

The reader is referred to~\citet{CPR2014, CPR2016} for an in-depth treatise on this subject, summarized next.

\begin{Def}[Hypergraph]
A \emph{hypergraph} $\H$ is a pair $(V,\A)$, where $V$ is the set of nodes, and $\A$ is the set of \emph{hyperarcs}.
Each hyperarc $A=(t_A, H_A, w_A)\in \A$ has a distinguished node $t_A$ called the \emph{tail} of $A$,
and a nonempty set $H_A\subseteq V\setminus\{t_A\}$ containing the \emph{heads} of $A$;
to each head $v\in H_A$ is associated a \emph{weight} $w_A(v)\in\Z$.
\end{Def}

Provided that $|A| \triangleq |H_A\cup \{t_A\}|$, the \emph{size} of a hypergraph $\H = (V,\A)$
is defined as $m_{\A}\triangleq \sum_{A\in\A}|A|$; it is used as a measure for the encoding length of $\H$.
If $|A|=2$, then $A=(u, v, w)$ can be regarded as a \emph{standard arc}. In this way, hypergraphs generalize graphs.

A \HTN is a weighted hypergraph $\H=(V,\A)$ where a node represents a \emph{time-point} to be scheduled,
and a hyperarc represents a disjunction of temporal distance \emph{constraints} between the \emph{tail} and the \emph{heads}.

In the \HTN framework the consistency problem is the following decision problem.
\begin{Def}[\HTNC]
Given some \HTN \mbox{$\H=(V,\A)$}, decide whether there is a schedule \mbox{$\phi:V \rightarrow \RR$} such that:
\[\phi(t_A) \geq \min_{v\in H_A} \{ \phi(v) - w_A(v) \}, \;\forall\; A\in\A\]
any such a schedule \mbox{$\phi:V \rightarrow \RR$} is called \textit{feasible}.

A \HTN is called \textit{consistent} whenever it admits at least one feasible schedule.
The problem of checking whether a given \HTN is consistent is named \HTNC, and the following algorithmic result holds on it.
\end{Def}

\begin{Thm}{\citep{CPR2014}}\label{Teo:MainAlgorithms}
There exists an $O\big((|V|+|\A|) m_{\A} W\big)$ pseudo-polynomial time algorithm for checking \HTNC;
moreover, when the input \HTN $\H=(V, \A)$ is consistent,
the algorithm returns as output a feasible schedule $\phi:V\rightarrow \RR$ of $\H$.
Here, $W\triangleq \max_{A\in\A, v\in H_A} |w_A(v)|$ is the maximum absolute weight value.
\end{Thm}

\subsection{$\varepsilon$-Dynamic-Consistency}
In \CSTN{s}, decisions about the precise timing of actions are postponed until execution time,
when information gathered from the observation nodes can be taken into account.
However, the planner is allowed to factor in an outcome, and differentiate its strategy according to it,
only \emph{strictly} after the outcome has been
observed (whence the strict inequality in Definition~\ref{def:scenario_history}).
Notice that this definition does not take into account the reaction-time,
which, in most applications, is non-negligible.
In order to deliver algorithms that can also deal
with the \emph{reaction-time} $\varepsilon>0$ of the planner,
we introduced in~\citet{CR15} a refined notion of DC.

\begin{Def}[$\varepsilon$-Dynamic-Consistency]\label{def:epsilonconsistency}
Given any \CSTN $( V, A, L, \Ord, \Ord{V}, P )$ and any real number $\varepsilon\in (0, +\infty)$,
an ES $\sigma\in\S_{\Gamma}$ is \emph{$\varepsilon$-dynamic} if it satisfies all of the $H_\varepsilon\text{-constraints}$,
namely, for any two scenarios $s_1, s_2\in \Sigma_P$ and any time-point $u\in V^+_{s_1,s_2}$,
the ES $\sigma$ satisfies the following constraint, which is denoted by $H_{\varepsilon}(s_1;s_2;u)$:
\[ [\sigma(s_1)]_u \geq \min\Big(\{[\sigma(s_2)]_u\}\cup\{[\sigma(s_1)]_v + \varepsilon\mid v\in\Delta(s_1; s_2)\}\Big). \tag{$H_{\varepsilon}(s_1;s_2;u)$} \]
We say that a \CSTN $\Gamma$ is \emph{$\varepsilon$-dynamically-consistent ($\varepsilon$-DC)} if it admits $\sigma\in\S_{\Gamma}$
which is both viable and $\varepsilon$-dynamic.
\end{Def}
\begin{figure}[!htb]
\centering
	\begin{tikzpicture}[arrows=->,scale=0.8,node distance=0.4 and 1.5]
    		\node[node] (v1) {$v_1$};
 	 	\node[node,below=of v1] (v2) {$v_2$};
		\node[node,left=of v2] (u) {$u_{s_1}$};
		\node[node,below=of u, yshift=-3ex] (v3) {$u_{s_2}$};

\node[below left=of u, yshift=.85ex, xshift=-10ex] (fakeL) {};

\node[below right=of v2, yshift=.85ex, xshift=10ex] (fakeR) {};
\node[above=of fakeR, yshift=-1ex, xshift=-5ex] (fakeUp) {$s_1$};
\node[below=of fakeR, yshift=1ex, xshift=-5ex] (fakeRdown) {$s_2$};

\draw[dashed] (u) to node[timeLabel, above,sloped] {$A, -\varepsilon$} (v1);%
\draw[dashed] (u) to node[timeLabel, above,sloped] {$A, -\varepsilon$} (v2);%
\draw[dashed] (u) to node[timeLabel, above,sloped] {$A, 0$} (v3);%

\draw[>=, dotted] (fakeL) to node[timeLabel,above,sloped] {} (fakeR);%

	\end{tikzpicture}
\caption{An $H_{\varepsilon}(s_1;s_2;u)$ constraint, modeled as a hyperarc.}\label{fig:H_epsilon_constraint}
\end{figure}
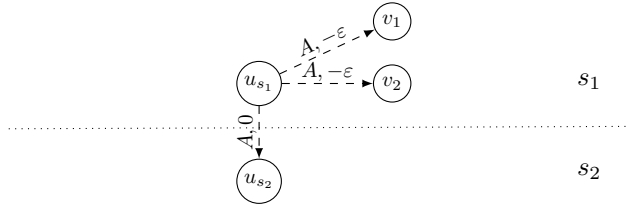
As shown in~\citet{CR15}, $\varepsilon$-DC can be modeled in terms of \HTNC.
\figref{fig:H_epsilon_constraint} depicts an illustration of an $H_{\varepsilon}(s_1;s_2;u)$ constraint,
modeled as an hyperarc.

Also,~\citet{CR15} proved that DC coincides with $\hat\varepsilon$-DC,
provided that $\hat\varepsilon\triangleq |\Sigma_P|^{-1}|V|^{-1}$.
\begin{Thm}\label{thm:epsilonconsistency} Let $\hat\varepsilon\triangleq |\Sigma_P|^{-1}|V|^{-1}$.
Then, $\Gamma$ is DC if and only if $\Gamma$ is $\hat\varepsilon$-DC.
Moreover, if $\Gamma$ is $\varepsilon$-DC for some $\varepsilon>0$,
then $\Gamma$ is $\varepsilon'$-DC for every $\varepsilon'\in (0, \varepsilon]$.
\end{Thm}
Then, the main result offered in~\citet{CR15} is a (pseudo) singly-exponential time DC-checking and $\epsilon$-DC-checking
	procedure (which is ultimately based on the consistency checking of \HTN{s}, \ie Theorem~\ref{Teo:MainAlgorithms}).
\begin{Thm}\label{thm:mainresult}
	The following two propositions hold. Here, $W\triangleq \max_{a\in A} |w_a|$ is the maximum absolute weight of $\Gamma$.
\begin{enumerate}
\item There exists an $O(|\Sigma_P|^{2}|A|^2 +|\Sigma_P|^3|A||V||P| + |\Sigma_P|^4|V|^2|P|)WD$ time algorithm
deciding \eDCC on input $(\Gamma, \epsilon)$,
for any \CSTN $\Gamma=( V, A, L, \Ord, \Ord{V}, P )$ and any rational number $\epsilon=N/D$ where $N,D\in \N_0$.
In particular, given any $\epsilon$-dynamically-consistent \CSTN $\Gamma$,
the algorithm returns as output a viable and $\epsilon$-dynamic execution strategy for $\Gamma$.
\item There exists an $O(|\Sigma_P|^{3}|A|^2|V| +
	|\Sigma_P|^4|A||V|^2|P| + |\Sigma_P|^5|V|^3|P|)W$ time algorithm for checking \DCC
on any input $\Gamma=( V, A, L, \Ord, \Ord{V}, P )$.
In particular, given any dynamically-consistent \CSTN $\Gamma$,
the algorithm returns a viable and dynamic execution strategy for $\Gamma$.
\end{enumerate}
\end{Thm}

\section{DC with Instantaneous Reaction-Time}\label{sect:Algo}
Theorem~\ref{thm:epsilonconsistency} points out the equivalence between $\varepsilon$-DC and DC,
that arises for a sufficiently small $\varepsilon>0$. However, Definition~\ref{def:epsilonconsistency}
makes sense even if $\varepsilon=0$, so a natural question is what happens
to the above mentioned relationship between DC and $\varepsilon$-DC when $\varepsilon=0$.
In this section we first show that $0$-DC doesn't imply DC and, moreover,
that $0$-DC is in itself too weak to capture an adequate
notion of DC with an instantaneous reaction-time (in the sense that a $0$-dynamic ES needs not come in the form of a tree-like structure).
In light of this we will introduce a stronger notion,
which is named \emph{ordered-Dynamic-Consistency (\oDCC)};
this will turn out to be an adequate (tree-like) notion of DC with an instantaneous reaction-time.
Let us provide next an example of a \CSTN $\Gamma_{\Box}$ which is $0$-DC but not DC.
\begin{Ex}[\CSTN $\Gamma_{\Box}$]\label{EX:CSTN_box}
Consider the following \CSTN: \[\Gamma_{\Box}=(V_\Box, A_\Box, L_\Box, \Ord_\Box, \Ord{V}_\Box, P_\Box),\]
where:
\begin{figure}[!htb]
\centering
\begin{tikzpicture}[arrows=->,scale=0.65,node distance=4 and 2.5]
	\node[node, scale=1.3, label={above:$\sizedcircled{.75}{\bf 1}$}] (one) {$\top$};
	\node[node, scale=1.3, below =of one, xshift=0ex, yshift=0ex,
		label={above, xshift=-1ex, yshift=-.25ex : \footnotesize $b?$}] (beta) {$B$};
	\node[node, scale=1.3, left=of beta, xshift=-15ex, label={above : \footnotesize $a?$}] (alpha) {$A$};
	\node[node, scale=1.3, right=of beta, xshift=15ex, label={above : \footnotesize $c?$}] (gamma) {$C$};
 	\node[node, scale=1.3, label={below:$\sizedcircled{.75}{\bf 0}$}, below=of beta, yshift=0ex] (zero) {$\bot$};
	 \node[xshift=12ex, left = of alpha] (fakeL) {};
 	 \node[xshift=-12ex, right = of gamma] (fakeR) {};
	\draw[-] (zero) to [bend left=55] node[xshift=-12.5ex, yshift=16ex,above] {$+1$} (fakeL.north);
	\draw[] (fakeL) to [bend left=55] node[xshift=0ex, yshift=0ex,above] {} (one);
	\draw[-] (one) to [bend left=55] node[xshift=12.5ex, yshift=-18ex, above] {$-1$} (fakeR.south);
	\draw[] (fakeR) to [bend left=55] node[xshift=0ex, yshift=0ex, above] {} (zero);
	\draw[] (zero) to [bend left=25] node[xshift=-3.25ex, yshift=2ex,below] {\footnotesize $0,\neg b$} (alpha);
	\draw[] (alpha) to [bend left=0] node[xshift=1ex, yshift=-.5ex,below] {\footnotesize $0$} (zero);
	\draw[] (zero) to [bend right=25] node[above] {\footnotesize $0,c$} (alpha);
	\draw[] (alpha) to [] node[xshift=-6ex, below] {\footnotesize $0, b \neg c$} (one.south);
	\draw[] (zero) to [bend left=25] node[xshift=-.4ex, yshift=-.75ex, above] {\footnotesize $0,\neg a$} (beta);
	\draw[] (beta) to [bend left=0] node[below, xshift=.6ex] {\footnotesize $0$} (zero);
	\draw[] (zero) to [bend right=25] node[above, xshift=.4ex, yshift=.75ex] {\footnotesize $0,\neg c$} (beta);
	\draw[] (beta) to [] node[xshift=.2ex, yshift=-3ex, above] {\footnotesize $0, ac$} (one.south);
	\draw[] (zero) to [bend left=25] node[above, xshift=1ex, yshift=.75ex] {\footnotesize $0,a$} (gamma);
	\draw[] (gamma) to [bend right=0] node[below] {\footnotesize $0$} (zero);
	\draw[] (zero) to [bend right=25] node[above, xshift=3.8ex, yshift=-.5ex] {\footnotesize $0,b$} (gamma);
	\draw[] (gamma) to [] node[xshift=6.5ex, below] {\footnotesize $0, \neg a \neg b$} (one.south);
\end{tikzpicture}
\caption{A \CSTN $\Gamma_\Box$ which is $0$-DC but not DC.}
\label{FIG:cstn_box}
\end{figure}
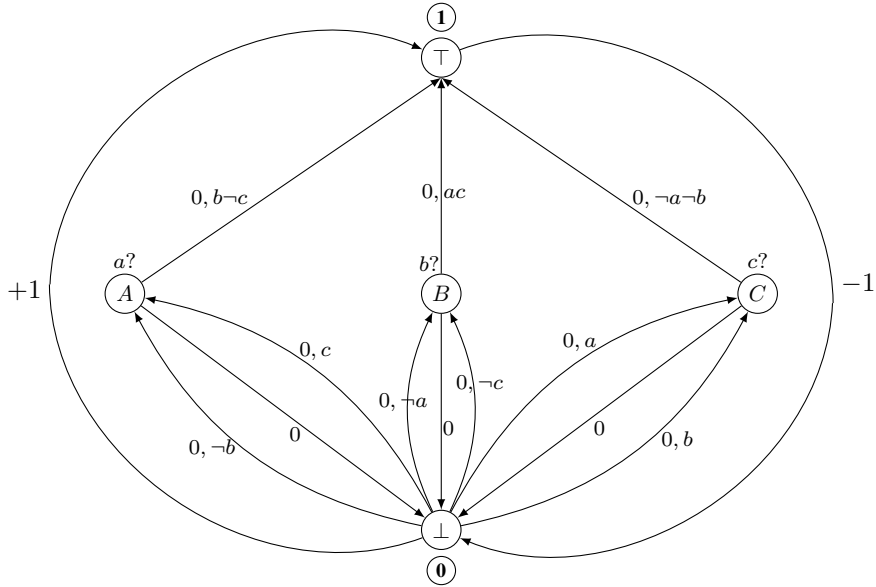

The intuition underlying $\Gamma_{\Box}$ is that of considering three observation time-points, $A$, $B$ and $C$ (observing propositional letters $a,b,c$, respectively),
each of which must be scheduled either at time $0$ or at time $1$
	depending on the propositional outcome of the other two propositional letters (\ie $A$ depending on $b,c$; $B$ on $a,c$; $C$ on $a,b$).
The idea is that of choosing the labels so that to introduce sort of an \emph{instantaneous circularity} condition in the corresponding temporal constraints, which
can be resolved in case of $\epsilon=0$, but not in case of any $\epsilon>0$ no matter how small.
For this, the prescribed rules will be the following:
 $A$ must be scheduled at time 0 either if $b$ is false or $c$ is true, otherwise $A$ must be scheduled at time $1$;
 $B$ must be scheduled at time 0 either if $a$ is false or $c$ is false, otherwise $C$ must be scheduled at time $1$;
 finally, $C$ must be scheduled at time 0 either if $a$ is true or $b$ is true,
 	otherwise $C$ is scheduled at time $1$.

Formally, this can be encoded as follows:

-- $V_\Box = \{ \bot, \top, A, B, C \}$;

-- $A_\Box = \{ (\top-\bot\leq 1, \lambda), (\bot-\top\leq -1, \lambda),
			(\top-A\leq 0, b\wedge \neg c), (\top-B\leq 0, a\wedge c),
			(\top-C\leq 0, \neg a \wedge \neg b), (\bot-A\leq 0, \lambda), (A-\bot\leq 0, \neg b), (A-\bot\leq 0, c),
			(\bot-B\leq 0, \lambda), (B-\bot\leq 0, \neg a), (A-\bot\leq 0, \neg c),
			(\bot-C\leq 0,\lambda), (C-\bot\leq 0, a), (C-\bot\leq 0, b) \}$;

-- $L_\Box(A)=L_\Box(B)=L_\Box(C)=L_\Box(\bot)=L_\Box(\top)= \lambda$;

-- $\Ord_\Box(a)= A$, $\Ord_\Box(b)= B$, $\Ord_\Box(c)= C$;

-- $\Ord{V}_\Box= \{A,B,C\}$;

-- $P_\Box = \{a,b,c\}$.

See \figref{FIG:cstn_box} for an illustration of $\Gamma_\Box$.
\end{Ex}

\begin{Prop}
The \CSTN $\Gamma_{\Box}$ (Example~\ref{EX:CSTN_box}, \figref{FIG:cstn_box}) is $0$-DC.
\end{Prop}
\begin{proof}
Consider the execution strategy $\sigma_{\Box} : \Sigma_{P_\Box}\rightarrow \Psi_{V_\Box}$:

-- $[\sigma_{\Box}(s)]_A\triangleq s(a\wedge b\wedge \neg c) + s(\neg a\wedge b\wedge \neg c)$;

-- $[\sigma_{\Box}(s)]_B\triangleq s(a\wedge b\wedge c) + s(a\wedge \neg b\wedge c)$;

-- $[\sigma_{\Box}(s)]_C\triangleq s(\neg a\wedge \neg b\wedge \neg c) + s(\neg a\wedge \neg b\wedge c)$;

-- $[\sigma_{\Box}(s)]_\bot\triangleq 0$ and $[\sigma_{\Box}(s)]_\top\triangleq 1$, for every $s\in\Sigma_P$.

An illustration of $\sigma_{\Box}$ is offered in Fig~\ref{FIG:ES_cstn_box}.
Three cubical graphs are depicted in which every node is labelled with two coordinates (\ie vertex and scenario)
	as $v_s=(v,s)$ for some $(v,s)\in V_\Box\times\Sigma_{P_\Box}$:
an edge connects ${v_1}_{s_1}$ and ${v_2}_{s_2}$ if and only if:
(i) the vertex-coordinate is the same \ie $v_1=v_2$ and (ii) the
Hamming distance between the scenario-coordinates $s_1$ and $s_2$ is unitary;
each scenario $s\in\Sigma_{P_\Box}$ is represented as $s=\alpha\beta\gamma$ for $\alpha,\beta,\gamma\in \{0,1\}$,
where $s(a)=\alpha$, $s(b)=\beta$, $s(c)=\gamma$; moreover, each node $v_s=(v,s)\in V_{\Box}\times\Sigma_{P_{\Box}}$
is filled in black if $[\sigma_{\Box}(s)]_v=0$, and in white if $[\sigma_{\Box}(s)]_v=1$.
So the underlying intuition is that all three $3$-cubes own both black and white nodes, but each of them,
in its own dimension, decomposes into two identically colored $2$-cubes.
\begin{figure}[!htb]
\centering
\begin{tikzpicture}[arrows=-,scale=1.2]
	\node[node, fill=black!50, scale=.6, label={below left, xshift=1.4ex, yshift=.5ex : $A_{000}$}] (A1) {};
	\node[node, scale=.6, right = of A1, label={below left, xshift=1.4ex, yshift=.5ex: $A_{010}$}]  (A2) {};
	\node[node, fill=black!50, scale=.6, above = of A2, label={below left, xshift=1.4ex, yshift=.9ex : $A_{011}$}] (A3) {};
	\node[node, fill=black!50, scale=.6, above =of A1, label={below left, xshift=1.4ex, yshift=.9ex : $A_{001}$}] (A4) {};
	\node[node, fill=black!50, scale=.6, xshift=5ex, yshift=7ex,
				label={right, xshift=-1.2ex, yshift=-1.1ex : $A_{100}$}] (A5) {};
	\node[node, scale=.6, right = of A5, label={right, xshift=-2ex, yshift=-2ex: $A_{110}$}] (A6) {};
	\node[node, fill=black!50, scale=.6, above = of A6, label={above left, xshift=0ex, yshift=-1ex : $A_{111}$}] (A7) {};
	\node[node, fill=black!50, scale=.6, above =of A5, label={above left, xshift=0ex, yshift=-1ex : $A_{101}$}] (A8) {};
	\draw[] (A1) to [] node[] {} (A2);
	\draw[] (A2) to [] node[] {} (A3);
	\draw[] (A3) to [] node[] {} (A4);
	\draw[] (A4) to [] node[] {} (A1);
	\draw[dashed] (A5) to [] node[] {} (A6);
	\draw[] (A6) to [] node[] {} (A7);
	\draw[] (A7) to [] node[] {} (A8);
	\draw[dashed] (A8) to [] node[] {} (A5);
	\draw[dashed] (A1) to [] node[] {} (A5);
	\draw[] (A2) to [] node[] {} (A6);
	\draw[] (A3) to [] node[] {} (A7);
	\draw[] (A8) to [] node[] {} (A4);
	\node[node, fill=black!50, right=of A2, xshift=1.5ex, scale=.6,
		label={below left, xshift=1.4ex, yshift=.5ex : $B_{000}$}] (B1) {};
	\node[node, fill=black!50, scale=.6, right = of B1, label={below left, xshift=1.4ex, yshift=.5ex: $B_{010}$} ] (B2) {};
	\node[node, fill=black!50, scale=.6, above = of B2, label={below left, xshift=1.4ex, yshift=.9ex : $B_{011}$}] (B3) {};
	\node[node, fill=black!50, scale=.6, above =of B1,  label={below left, xshift=1.4ex, yshift=.9ex : $B_{001}$}] (B4) {};
	\node[node, fill=black!50, scale=.6, right = of A2, xshift=6.5ex,
		yshift=7ex, label={right, xshift=-.8ex, yshift=-1.1ex : $B_{100}$} ] (B5) {};
	\node[node, fill=black!50, scale=.6, right = of B5,  label={right, xshift=-2ex, yshift=-2.2ex : $B_{110}$}  ] (B6) {};
	\node[node, scale=.6, above = of B6, label={above left, xshift=0ex, yshift=-1ex : $B_{111}$} ] (B7) {};
	\node[node, scale=.6, above =of B5, label={above left, xshift=0ex, yshift=-1ex : $B_{101}$} ] (B8) {};
	\draw[] (B1) to [] node[] {} (B2);
	\draw[] (B2) to [] node[] {} (B3);
	\draw[] (B3) to [] node[] {} (B4);
	\draw[] (B4) to [] node[] {} (B1);
	\draw[dashed] (B5) to [] node[] {} (B6);
	\draw[] (B6) to [] node[] {} (B7);
	\draw[] (B7) to [] node[] {} (B8);
	\draw[dashed] (B8) to [] node[] {} (B5);
	\draw[dashed] (B1) to [] node[] {} (B5);
	\draw[] (B2) to [] node[] {} (B6);
	\draw[] (B3) to [] node[] {} (B7);
	\draw[] (B8) to [] node[] {} (B4);
	\node[node, right =of B2, xshift=1.5ex, scale=.6, label={below left, xshift=1.4ex, yshift=.5ex : $C_{000}$} ] (C1) {};
	\node[node, fill=black!50,scale=.6, right = of C1, label={below left, xshift=1.4ex, yshift=.5ex: $C_{010}$} ] (C2) {};
	\node[node, fill=black!50,scale=.6, above = of C2, label={below left, xshift=1.4ex, yshift=.9ex : $C_{011}$} ] (C3) {};
	\node[node, scale=.6, above=of C1, label={below left, xshift=1.4ex, yshift=.9ex : $C_{001}$} ] (C4) {};
	\node[node, fill=black!50,scale=.6, right = of B2, xshift=6.5ex,
		yshift=7ex, label={right, xshift=-.8ex, yshift=-1.1ex : $C_{100}$} ] (C5) {};
	\node[node, fill=black!50,scale=.6, right = of C5, label={right, xshift=-2ex, yshift=-2.2ex: $C_{110}$}  ] (C6) {};
	\node[node, fill=black!50,scale=.6, above = of C6, label={above left, xshift=0ex, yshift=-1ex : $C_{111}$} ] (C7) {};
	\node[node, fill=black!50,scale=.6, above =of C5, label={above left, xshift=0ex, yshift=-1ex : $C_{101}$} ] (C8) {};
	\draw[] (C1) to [] node[] {} (C2);
	\draw[] (C2) to [] node[] {} (C3);
	\draw[] (C3) to [] node[] {} (C4);
	\draw[] (C4) to [] node[] {} (C1);
	\draw[dashed] (C5) to [] node[] {} (C6);
	\draw[] (C6) to [] node[] {} (C7);
	\draw[] (C7) to [] node[] {} (C8);
	\draw[dashed] (C8) to [] node[] {} (C5);
	\draw[dashed] (C1) to [] node[] {} (C5);
	\draw[] (C2) to [] node[] {} (C6);
	\draw[] (C3) to [] node[] {} (C7);
	\draw[] (C8) to [] node[] {} (C4);
\end{tikzpicture}
\caption{The ES $\sigma_{\Box}$ for the \CSTN $\Gamma_\Box$.}
\label{FIG:ES_cstn_box}
\end{figure}
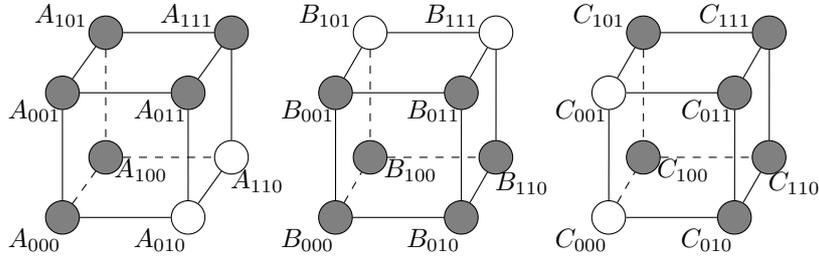
\figref{FIG:ES_plus_cstn_box} offers another visualization of $\sigma_{\Box}$
in which every component of the depicted graph corresponds to a restriction
\STN ${\Gamma^+_{\Box}}_s$ for some $s\in\Sigma_{P_\Box}$, where two scenarios $s_i, s_j\in \Sigma_{P_\Box}$ are
grouped together whenever ${\Gamma^+_\Box}_{s_i}={\Gamma^+_\Box}_{s_j}$.
Also note that \figref{FIG:ES_plus_cstn_box} represents each scenario as a triplet of truth values
	 $\big(s(a), s(b), s(c)\big)$, \eg 111, 110, etc (instead of triplet of literals $abc, ab\neg c$, etc).
It is easy to see from \figref{FIG:ES_plus_cstn_box} that $\sigma_{\Box}$ is viable for $\Gamma_\Box$.
In order to check that $\sigma_{\Box}$ is $0$-dynamic, look again at \figref{FIG:ES_plus_cstn_box},
and notice that for every $s_i, s_j\in \Sigma_{\Box}$, where $s_i\neq s_j$, there exists a time-point
$X\in\{A,B,C\}$ such that $[\sigma_{\Box}(s_i)]_X=0=[\sigma_{\Box}(s_j)]_X$ and $s_i(X)\neq s_j(X)$.
With this in mind it is easy to check that all of the $H_0$ constraints are thus satisfied by $\sigma_{\Box}$.
Therefore, the \CSTN $\Gamma_\Box$ is $0$-DC.
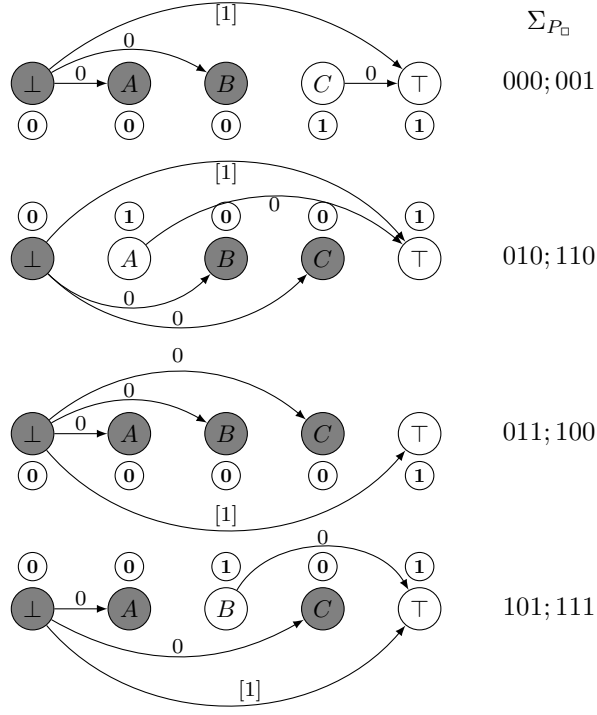
\begin{figure}[!htb]
\centering
\begin{tikzpicture}[arrows=->,scale=0.7,node distance=2.5 and 1]
	\node[node, scale=1.3, label={below:\sizedcircled{.75}{$\bf 1$}}] (one1) {$\top$};
	\node[node, scale=1.3, label={below:\sizedcircled{.75}{$\bf 1$}}, left=of one1] (gamma1) {$C$};
	\node[node, fill=black!50, scale=1.3, label={below:\sizedcircled{.75}{$\bf 0$}}, left=of gamma1] (beta1) {$B$};
	\node[node, fill=black!50, scale=1.3, label={below:\sizedcircled{.75}{$\bf 0$}}, left=of beta1] (alpha1) {$A$};
 	\node[node, fill=black!50, scale=1.3, left=of alpha1, label={below:\sizedcircled{.75}{$\bf 0$}}] (zero1) {$\bot$};
	\node[right=of one1, yshift=5ex] (fake1) {$\Sigma_{P_\Box}$};
	\node[right=of one1, xshift=-2ex] (s1) {$000; 001$};

	\node[node, scale=1.3, below=of one1, label={above:\sizedcircled{.75}{$\bf 1$}}] (one2) {$\top$};
	\node[node, fill=black!50, scale=1.3, label={above:\sizedcircled{.75}{$\bf 0$}}, left=of one2] (gamma2) {$C$};
	\node[node, fill=black!50, scale=1.3, label={above:\sizedcircled{.75}{$\bf 0$}}, left=of gamma2] (beta2) {$B$};
	\node[node, scale=1.3, label={above:\sizedcircled{.75}{$\bf 1$}}, left=of beta2] (alpha2) {$A$};
 	\node[node, fill=black!50, scale=1.3, left=of alpha2, label={above:\sizedcircled{.75}{$\bf 0$}}] (zero2) {$\bot$};
	\node[right=of one2, xshift=-2ex] (s2) {$010; 110$};

	\node[node, scale=1.3, below=of one2, label={below:\sizedcircled{.75}{$\bf 1$}}] (one3) {$\top$};
	\node[node, fill=black!50, scale=1.3, label={below:\sizedcircled{.75}{$\bf 0$}}, left=of one3] (gamma3) {$C$};
	\node[node, fill=black!50, scale=1.3, label={below:\sizedcircled{.75}{$\bf 0$}}, left=of gamma3] (beta3) {$B$};
	\node[node, fill=black!50, scale=1.3, label={below:\sizedcircled{.75}{$\bf 0$}}, left=of beta3] (alpha3) {$A$};
 	\node[node, fill=black!50, scale=1.3, left=of alpha3, label={below:\sizedcircled{.75}{$\bf 0$}}] (zero3) {$\bot$};
	\node[right=of one3, xshift=-2ex] (s3) {$011; 100$};

	\node[node, scale=1.3, below=of one3, label={above:\sizedcircled{.75}{$\bf 1$}}] (one4) {$\top$};
	\node[node, fill=black!50, scale=1.3, label={above:\sizedcircled{.75}{$\bf 0$}}, left=of one4] (gamma4) {$C$};
	\node[node, scale=1.3, label={above:\sizedcircled{.75}{$\bf 1$}}, left=of gamma4] (beta4) {$B$};
	\node[node, fill=black!50, scale=1.3, label={above:\sizedcircled{.75}{$\bf 0$}}, left=of beta4] (alpha4) {$A$};
 	\node[node, fill=black!50, scale=1.3, left=of alpha4, label={above:\sizedcircled{.75}{$\bf 0$}}] (zero4) {$\bot$};
	\node[right=of one4, xshift=-2ex] (s4) {$101; 111$};

	\draw[] (zero1) to [bend left=0] node[xshift=0ex, yshift=-.5ex, above] {\footnotesize $0$} (alpha1);
	\draw[] (zero1) to [bend left=30] node[xshift=0ex, yshift=-.5ex, above] {\footnotesize $0$} (beta1);
	\draw[] (zero1) to [bend left=40] node[xshift=0ex, yshift=.7ex, below] {\footnotesize $[1]$} (one1);
	\draw[] (gamma1) to [bend left=0] node[xshift=0ex, yshift=-.5ex, above] {\footnotesize $0$} (one1);

	\draw[] (zero2) to [bend left=50] node[xshift=0ex, yshift=.7ex, below] {\footnotesize $[1]$} (one2);
	\draw[] (zero2) to [bend right=45] node[xshift=0ex, yshift=-.6ex, above] {\footnotesize $0$} (beta2);
	\draw[] (zero2) to [bend right=45] node[xshift=0ex, yshift=-.6ex, above] {\footnotesize $0$} (gamma2);
	\draw[] (alpha2) to [bend left=40] node[xshift=0ex, yshift=.7ex, below] {\footnotesize $0$} (one2);

	\draw[] (zero3) to [bend left=0] node[xshift=0ex, yshift=-.5ex, above] {\footnotesize $0$} (alpha3);
	\draw[] (zero3) to [bend left=30] node[xshift=0ex, yshift=-.5ex, above] {\footnotesize $0$} (beta3);
	\draw[] (zero3) to [bend right=50] node[xshift=0ex, yshift=-.5ex, above] {\footnotesize $[1]$} (one3);
	\draw[] (zero3) to [bend left=40] node[xshift=0ex, yshift=0ex, above] {\footnotesize $0$} (gamma3);

	\draw[] (zero4) to [bend left=0] node[xshift=0ex, yshift=-.5ex, above] {\footnotesize $0$} (alpha4);
	\draw[] (beta4) to [bend left=60] node[xshift=0ex, yshift=-.5ex, above] {\footnotesize $0$} (one4);
	\draw[] (zero4) to [bend right=50] node[xshift=2ex, yshift=-.5ex, above] {\footnotesize $[1]$} (one4);
	\draw[] (zero4) to [bend right=30] node[xshift=0ex, yshift=-.5ex, above] {\footnotesize $0$} (gamma4);
\end{tikzpicture}
\caption{The restrictions ${\Gamma^+_{\Box}}_s$ for $s\in\Sigma_{P_\Box}$,
where the execution times $[\sigma_{\Box}(s)]_v\in \{\bf 0,1\}$ are depicted in a circled bold face.}
\label{FIG:ES_plus_cstn_box}
\end{figure}
\end{proof}

\begin{Prop}\label{prop:gamma_box_not_dc}
The \CSTN $\Gamma_\Box$ is not DC.
\end{Prop}
\begin{proof}
Let $\sigma$ be a viable ES for $\Gamma_\Box$. Then, $\sigma$ must be the ES $\sigma_\Box$ depicted in \figref{FIG:ES_plus_cstn_box},
there is no other choice here. Let $\hat{s}\in \Sigma_{P_\Box}$. Then, it is easy to check from \figref{FIG:ES_plus_cstn_box} that:
(i) $[\sigma_{\Box}(\hat{s})]_\bot=0$, $[\sigma_{\Box}(\hat{s})]_\top=1$,
	and it holds $[\sigma_{\Box}(\hat{s})]_X \in \{0,1\}$ for every $X\in \{A,B,C\}$;
(ii) there exists at least two observation time-points $X\in \{A,B,C\}$ such that $[\sigma(\hat{s})]_X = 0$; still,
(iii) there is no $X\in \{A,B,C\}$ such that $[\sigma(s)]_X = 0$ for every $s\in \Sigma_{P_{\Box}}$,
	\ie no observation time-point is executed first at all possible scenarios. Therefore, the ES $\sigma_{\Box}$ is not dynamic.
\end{proof}

In order to tackle on the tricky situations just observed above,
	we now introduce a stronger notion of dynamic consistency;
	it is named \emph{ordered-Dynamic-Consistency (\oDCC)},
and it takes explicitly into account an additional ordering between the observation time-points scheduled at the same execution time.
Basically, this will preserve the semantics of DC, but
meanwhile ruling out the possibility of crafty constructions
	like the $\Gamma_\Box$ of Example~\ref{EX:CSTN_box}.

\begin{Def}[$\pi$-Execution-Strategy]\label{def:executionstrategy}
An \emph{ordered-Execution-Strategy ($\pi$-ES)} for $\Gamma$ is a mapping:
\[ \sigma : s\mapsto ([\sigma(s)]^t, [\sigma(s)]^\pi), \]
where $s\in\Sigma_P$, $[\sigma(s)]^t\in\Phi_V$, and finally,
$[\sigma(s)]^{\pi}: \Ord{V}^+_s \rightleftharpoons \{1, \ldots, |\Ord{V}^+_s|\}$ is bijective.
The set of $\pi$-ES of $\Gamma$ is denoted by $\S_{\Gamma}$.
For any $s\in\Sigma_P$, the \emph{execution time} of a time-point $v\in V^+_s$ in the schedule
$[\sigma(s)]^t\in\Phi_{V^+_s}$ is denoted by $[\sigma(s)]^t_v\in\RR$;
the \emph{position} of an observation $\Ord_p\in {\Ord}V^+_s$ in
$\sigma(s)$ is $[\sigma(s)]^{\pi}_{\Ord_p}$.

We require positions to be \emph{coherent} \wrt execution times,
\ie \[\forall (\Ord_p,\Ord_q\in \Ord{V}^+_s)\text{ if }[\sigma(s)]^{t}_{\Ord_p} < [\sigma(s)]^{t}_{\Ord_q},
\text{ then } [\sigma(s)]^{\pi}_{\Ord_p} < [\sigma(s)]^{\pi}_{\Ord_q}.\] In addition,
it is worth adopting the notation:
	\[[\sigma(s)]^\pi_v\triangleq |\Ord{V}|+1, \text{ whenever } v\in V^+_s\setminus \Ord{V}.\]
\end{Def}

\begin{Rem}
The above definition of $\pi$-ES also allows us to relax the \emph{WD2} assumption for \CSTN{s}, as follows.

(\emph{WD2}') For each $p\in P$ and each $u\in V$ such that either $p$ or $\neg p$ appears in $L(u)$,
 we require:
$\Sub(L(u), L(\Ord_p))$, and $( \Ord_p-u\leq 0, L(u)) \in A$;
plus, in case $u\in \Ord V$ is also an observation time-point and it is scheduled at the same instant of time of $\Ord_p$,
we require $u$ to be placed \emph{after} $\Ord_p$ in the inner ordering relation of the corresponding $\pi$-ES.

Intuitively this means that, whenever a time-point label $L(u)$
contains some propositional letter $p\in P$, and $u$ is eventually scheduled,
then $\Ord_p$ must be scheduled before or at the same instant of time of $u$,
 but still at a subsequent \emph{position} in the additional ordering
 that is induced by the $\pi$-DC.
\end{Rem}

Also the notion of history (\ie Definition~\ref{def:scenario_history}) can be extended in a natural way, by considering
both of the two coordinates $([\sigma(s)]^t, [\sigma(s)]^\pi)$.
\begin{Def}[$\pi$-History]\label{def:pi_scenario_history}
Let $\sigma\in\S_{\Gamma}$, $s\in\Sigma_P$, and let $\tau\in\RR$
and $\psi\in \{1, \ldots, |V|\}$. The \emph{ordered-history} $\oHst(\tau, \psi,s,\sigma)$ of
$\tau$ and $\psi$ in the scenario $s$, under the $\pi$-ES $\sigma$, is defined as:
\begin{align*} \oHst(\tau, \psi,s,\sigma)\triangleq \big\{ & \big(p, s(p)\big)\in P\times\{0,1\}
	\mid  \\
	& \Ord_p~\in~\Ord{V}^+_{s}, [\sigma(s)]^{t}_{\Ord_p}\leq \tau, [\sigma(s)]^{\pi}_{\Ord_p} < \psi \big\}.
\end{align*}
\end{Def}
We are finally in the position to define \oDCC.

\begin{Def}[$\pi$-Dynamic-Consistency]\label{def:pi-dc}
Any $\pi$-ES $\sigma\in \S_{\Gamma}$ is called \emph{$\pi$-dynamic} when,
for any two complete scenarios $s_1, s_2\in \Sigma_P$ and any time-point $v\in V^+_{s_1,s_2}$,
if $\tau\triangleq [\sigma(s_1)]^t_v$ and $\psi\triangleq [\sigma(s_1)]^\pi_v$, then:
\[\Con(\oHst(\tau, \psi, s_1, \sigma), s_2) \Rightarrow [\sigma(s_2)]^t_v = \tau, [\sigma(s_2)]^{\pi}_v = \psi.\]
We say that $\Gamma$ is \emph{$\pi$-dynamically-consistent (\oDCC)}
if it admits $\sigma\in\S_{\Gamma}$ which is both viable and $\pi$-dynamic.

The problem of checking whether a given \CSTN is \oDCC is named \emph{\oDCC}-Checking.
\end{Def}

\begin{Rem}\label{rem:tree_structure} Notice that,
due to the strict inequality ``$[\sigma(s)]^\pi_{\Ord_p}<\psi$" in
the definition of $\oHst(\cdot)$ (Definition~\ref{def:pi_scenario_history}), in a $\pi$-dynamic $\pi$-ES,
there must be exactly one $\Ord_{p'}\in\Ord{V}$,
for some $p'\in P$, which is executed at first
(\wrt both execution time and position) under all possible scenarios $s\in\Sigma_P$.
There is always a \emph{root} for a tree-like strategy.

Indeed, if $\tau$ is earliest time at which a strategy $\sigma$ executes some
non-zero time-point, then $\oHst(\tau,0,s,\sigma)=\emptyset$ for each
scenario $s\in\Sigma_P$; so, if $[\sigma(s)]_X = \tau$ is executed first in scenario
$s$, then it must be executed first in every scenario; otherwise, $\sigma$ is not $\pi$-dynamic.
\end{Rem}

\begin{Prop}The \CSTN $\Gamma_\Box$ is not \oDCC.\end{Prop}
\begin{proof}
The proof goes almost in the same way as that of Proposition~\ref{prop:gamma_box_not_dc}.
In particular, no observation time-point is executed first (\ie at time $t=0$ and position $\psi=1$) under all possible scenarios.
Since there is no first-in-time observation time-point, then, the ES $\sigma$ is not $\pi$-dynamic.
\end{proof}

We provide next a rather simple \CSTN $\Gamma_\pi$ which is \oDCC but not DC.
The underlying intuition being, to consider~just~one observation time-point $\Ord_p$,
which needs to be scheduled first (say, at time $0$), plus one more time-point~$X$ which
must be scheduled either at time $0$ (at the same time of the observation $\Ord_p$)
if $p$ is true, or at time $1$ otherwise.

See here below for the formal definition of $\Gamma_\pi$.

\begin{Ex} Define $\Gamma_\pi=(V_\pi, A_\pi, \Ord_\pi, \Ord{V}_\pi, P_\pi)$ as follows.
$V_\pi = \{\Ord_p, X, \top\}$, $A_\pi=\{(\top-\Ord_p\leq 1, \lambda), (\Ord_p-\top\leq -1, \lambda),
(X-\Ord_p\leq 0, p), (\top-X\leq 0, \neg p)\}$, $\Ord_\pi(p)=\Ord_p$, $\Ord{V_\pi}=\{\Ord_p\}$, $P_\pi=\{p\}$.
\figref{FIG:CSTN_pi} depicts the \CSTN $\Gamma_\pi$.
\begin{figure}[!htb]
\centering
\begin{tikzpicture}[arrows=->,scale=1,node distance=1 and 4]
	\node[node, scale=1.3, label={above:\footnotesize $p?$}] (P) {$\Ord_p$};
	\node[node, scale=1.3, below = of P, xshift=12ex] (X) {$X$};
	\node[node, scale=1.3, right = of P] (one) {$\top$};
	\draw[] (P) to [bend right=20] node[xshift=0ex, yshift=0ex, above] {$1$} (one);
	\draw[] (one) to [bend right=20] node[xshift=0ex, yshift=0ex, above] {$-1$} (P);
	\draw[] (P) to [bend left=0] node[xshift=-2ex, yshift=1ex, below] {$0, p$} (X);
	\draw[] (X) to [bend left=0] node[xshift=3ex, yshift=1ex, below] {$0, \neg p$} (one);
\end{tikzpicture}
\caption{The \CSTN $\Gamma_\pi$ that is $\pi$-DC but not DC.}
\label{FIG:CSTN_pi}
\end{figure}
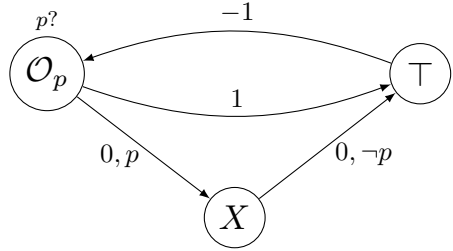
\end{Ex}
\begin{Prop}The \CSTN $\Gamma_\pi$ is \oDCC, but it is not DC.\end{Prop}
\begin{proof}
Let $s_1,s_2\in\Sigma_{P_\pi}$ be two scenarios such that $s_1(p)=1$ and $s_2(p)=0$.
Consider the $\pi$-ES $\sigma$ defined as follows:
$[\sigma(s_1)]^t_{\Ord_p}=[\sigma(s_1)]^t_{X}=0$, $[\sigma(s_1)]^t_{\top}=1$;
and $[\sigma(s_2)]^t_{\Ord_p}=0$, $[\sigma(s_2)]^t_{X}=[\sigma(s_2)]^t_{\top}=1$;
finally, $[\sigma(s)]^\pi_{\Ord_p}=1$, $[\sigma(s)]^\pi_{X}=[\sigma(s)]^\pi_{\top}=2$, for all $s\in \{s_1,s_2\}$.
Then, $\sigma$ is viable and $\pi$-dynamic for $\Gamma_\pi$. To see that $\Gamma_\pi$ is not DC, pick any $\varepsilon>0$.
Notice that any viable ES must schedule $X$ either at $t=0$ or $t=1$, depending on the outcome of $\Ord_p$,
which in turn happens at $t=0$; however, in any $\varepsilon$-dynamic strategy,
the planner can't react to the outcome of $\Ord_p$ before time $t=\varepsilon>0$.
This implies that $\Gamma_\pi$ is not $\varepsilon$-DC. Since $\varepsilon$ was chosen arbitrarily ($\varepsilon>0$),
then $\Gamma_\pi$ can't be DC by Theorem~\ref{thm:epsilonconsistency}.
\end{proof}
So $\Gamma_\pi$ is $\varepsilon$-DC for $\varepsilon=0$ but for \emph{no} $\varepsilon>0$.
In summary, the following relationships hold among
    the different versions of DC:
\begin{center}
\begin{tabular}{|c|}
\hline
[$\varepsilon$-DC, $\forall\varepsilon\in(0,\hat\varepsilon]$] ${\Leftrightarrow}$ DC $\stackrel{\not\Leftarrow}{\Rightarrow}$
	\oDCC $\overset{\not\Leftarrow}{\Rightarrow}$ [$\varepsilon$-DC, for $\varepsilon=0$] \\
\hline
\end{tabular}
\end{center}
where $\hat\varepsilon\triangleq |\Sigma_P|^{-1}\cdot |V|^{-1}$ as in Theorem~\ref{thm:epsilonconsistency}.

\subsection{The ps-tree: ``skeleton" structure for a $\pi$-dynamic $\pi$-ES}
In this subsection we introduce a labelled tree data structure, named \emph{ps-tree},
to capture the ``skeleton" ordered structure underlying any $\pi$-dynamic $\pi$-ES.
Basically, this is an outward (non-empty) rooted binary tree where nodes are labelled with propositions
and arcs with truth values, paths in the ps-tree represent execution scenarios and the set of letters along any path
represent precisely the set of observations that need to be scheduled along that path.
\begin{Def}[PS-Tree]\label{def:perm-tree} Let $P$ be any set of boolean variables.
	A permutation-scenario tree (ps-tree) $\pi_T$ over $P$ is an outward (non-empty) rooted binary tree such that:
\begin{itemize}
\item Each node $u$ of $\pi_T$ has an associated letter $p_u\in P$;
\item All the nodes that lie along a path leading from the root to a leaf are labelled with distinct letters from $P$.
\item Each arc $(u,v)$ of $\pi_T$ has an associated bit $b_{(u,v)}\in \{0,1\}$;
\item The two arcs $(u,v_l)$ and $(u,v_r)$ exiting
	a node $u$ can not have the same associated bit value, \ie \[b_{(u,v_l)}\neq b_{(u,v_r)}.\]
\end{itemize}
\end{Def}
\figref{fig:ex_ps-tree} depicts an example of a ps-tree, in which $P=\{a,b,c,d\}$. Note $\pi_T$ need not be a complete binary tree.
\begin{figure}[!htb]
\centering
\begin{forest}
for tree={circle, draw, grow=east, l sep=20pt, edge={->}}
[$a$,
    [$b$, edge label={node[midway, xshift=-1ex, yshift=-.6ex] {$0$}}
      [$c$, edge label={node[midway, yshift=-1ex] {$0$}} ]
      [$c$, edge label={node[midway, yshift=1ex] {$1$}}
		[$d$, edge label={node[midway, yshift=-1ex] {$0$}}]
		[$d$, edge label={node[midway, yshift=1ex] {$1$}}]
	]
    ]
    [$c$, edge label={node[midway, xshift=-1ex, yshift=.6ex] {$1$}}
      [$d$, edge label={node[midway, yshift=-1ex] {$0$}} ]
      [$b$, edge label={node[midway, yshift=1ex] {$1$}}
		[$d$, edge label={node[midway, yshift=-1ex] {$0$}}]
		[$d$, edge label={node[midway, yshift=1ex] {$1$}}]
	]
  ]
]
\end{forest}
\caption{An example of a ps-tree over $P=\{a,b,c,d\}$.}\label{fig:ex_ps-tree}
\end{figure}
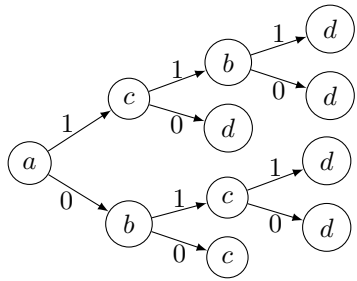
Next, the notion of ps-tree is refined to make it coherent \wrt a given $\pi$-dynamic \CSTN $\Gamma$.
\begin{Def}[$\pi_s$, $s_i$, Coherent-PS-Tree]\label{rem:ps-tree}
Let $\pi_T$ be a ps-tree over $P$, let $r$ be the root of $\pi_T$ and let $\gamma$ be any leaf.
Let $(r, v_2, \ldots, \gamma)$ be the sequence of the nodes encountered along the path going from $r$ down to $\gamma$ in $\pi_T$.

Note,
\begin{itemize}
\item The sequence of labels $\pi_\gamma=(p_r, p_{v_2}, \ldots, p_\gamma)$ is a
	permutation of the subset of letters $\{p_r, p_{v_2}, \ldots, p_\gamma\}\subseteq P$.
\item Each prefix sequence of bits $(b_{(r, v_2)}, \ldots, b_{(v_{i}, v_{i+1})})$,
	for each $i\in \{1,2,\ldots, k_{\gamma}-1\}$ (where $v_1\triangleq r$ and $v_{k_{\gamma}}\triangleq \gamma$),
	can be seen as a partial scenario $s_i$ over $P$;
	\ie $s_i(v_j)\triangleq b_{(v_j, v_{j+1})}$, for every $j\in\{1, \ldots, i\}$.
\end{itemize}
Then,
\begin{itemize}
\item $\pi_T$ is a \emph{coherent (c-ps-tree)} for \CSTN $\Gamma$ if,
	for every leaf $\gamma$ of $\pi_T$,
	\[\{\Ord_{p_r}, \Ord_{p_{v_2}}, \ldots, \Ord_{p_{\gamma}}\}=\Ord{V}^+_{s'}\] holds
	for every complete scenario $s'\in\Sigma_P$ such that $\Sub(s', s_{k_{\gamma}-1})$ is true.
\end{itemize}
\end{Def}

It is not difficult to see that a $\pi$-dynamic $\pi$-ES induces one and only one c-ps-tree $\pi_T$.
So, the existence of a suitable c-ps-tree is a necessary condition for a $\pi$-ES to be $\pi$-dynamic.
One may ask whether a $\pi$-dynamic $\pi$-ES can be reconstructed from its c-ps-tree;
the following subsection answers affirmatively.

\subsection{Verifying a c-ps-tree.}
This subsection builds on the notion of c-ps-tree to work out the details of the relationship between \oDCC and HyTN-Consistency.
Once this picture is in place, it will be easy to reduce to HyTN-Consistency the problem of deciding whether a given CSTN admits
a valid $\pi$-dynamic $\pi$-ES with a given c-ps-tree. This easy result already provides a first combinatorial algorithm for \oDCC,
though of doubly exponential complexity in $|P|$; a bound to be improved in later subsections,
but that can help sizing the sheer dimensionality and depth of the problem.

Firstly, the notion of \emph{Expansion} of \CSTN{s} is recalled from~\citet{CR15}.
Basically, this is a family of disjoint restriction \STN{s} $(V_s, A_s)$, one for each complete scenario $s\in\Sigma_P$;
it is obtained simply by projecting the initial \CSTN, according to a scenario $s$,
	on the corresponding restriction \STN,
and then by renaming (\ie renumbering)
	the time-points of the latter so that the whole family of \STN{s} will be vertex-disjoint.

\begin{Def}[Expansion $( V^{\text{Ex}}_{\Gamma}, \Lambda^{\text{Ex}}_{\Gamma})$]\label{def:expansion}
Consider a \CSTN $\Gamma=(V, A, L, \Ord, {\Ord}V, P)$. Consider the family of all (distinct) \STN{s} $(V_s, A_s)$,
one for each scenario $s\in\Sigma_P$, defined as follows:
\[V_s\triangleq\{v_s \mid v\in V^+_s\} \text{ and } A_s\triangleq\{(u_s,v_s,w) \mid (u,v,w) \in A^+_s\}.\]
The \emph{expansion} $( V^{\text{Ex}}_{\Gamma}, \Lambda^{\text{Ex}}_{\Gamma})$
of the \CSTN $\Gamma$ is defined as follows:
\[
( V^{\text{Ex}}_{\Gamma}, \Lambda^{\text{Ex}}_{\Gamma}
	) \triangleq \Big( \bigcup_{s\in\Sigma_P}V_s, \bigcup_{s\in\Sigma_P} A_s\Big ).
\]
\end{Def}

Notice, $( V^{\text{Ex}}_{\Gamma}, \Lambda^{\text{Ex}}_{\Gamma} )$ is an \STN with at most
$|V^{\text{Ex}}_{\Gamma}|\leq |\Sigma_P|\cdot |V|$ nodes and at most
$|\Lambda^{\text{Ex}}_{\Gamma}|\leq |\Sigma_P|\cdot |A|$ standard arcs.

We now show that the expansion of a \CSTN can be enriched with some standard arcs and
some hyperarcs in order to model \oDCC, by means of an \HTN denoted by $\H^{\pi_T}_{0}(\Gamma)$.
\begin{Def}[\HTN $\H^{\pi_T}_{0}(\Gamma)$]\label{def:Hepsilonzero}
Let $\Gamma=(V, A, L, \Ord, {\Ord}V, P)$ be a given \CSTN. Let $\pi_T$ be a given c-ps-tree over $P$.

Then, the \HTN $\H^{\pi_T}_{0}(\Gamma)$ is defined as follows:
\begin{itemize}
\item For every scenarios $s_1, s_2\in\Sigma_P$ and $u\in V^+_{s_1,s_2}\setminus \Ord{V}$,
define a hyperarc $\alpha=\alpha_{0}(s_1; s_2; u)$ as follows
(with the intention to model the constraint $H_{0}(s_1; s_2; u)$ as in Definition~\ref{def:epsilonconsistency}):
\[\alpha=\alpha_{0}(s_1; s_2; u)\triangleq ( t_\alpha, H_\alpha, w_\alpha), \]
where:
\begin{itemize}
\item $t_\alpha\triangleq u_{s_1}$ is the tail of the hyperarc $\alpha$;
\item $H_\alpha\triangleq \{u_{s_2}\}\cup \Delta(s_1; s_2)$ is the set of the heads;
\item $w_\alpha(u_{s_2})\triangleq 0$; $\forall (v\in\Delta(s_1; s_2))\; w_\alpha(v)\triangleq 0$.
\end{itemize}
\item[] Now, consider the expansion of the \CSTN $\Gamma$ $( V^{\text{Ex}}_{\Gamma}, \Lambda^{\text{Ex}}_{\Gamma})
= \big( \bigcup_{s\in\Sigma_P}V_s, \bigcup_{s\in\Sigma_P} A_s \big )$ (as in Definition~\ref{def:expansion}). Then:
\item For each internal node $x$ of $\pi_T$, $A'_x$ is a set of (additional) standard arcs defined as follows.

Let $\pi_x=(r, \ldots, x')$ be the sequence of all and only
the nodes along the path going from the root $r$ to the parent $x'$ of $x$ in $\pi_T$ (where we can assume $r'=r$).
Let $P_x\subseteq P$ be the corresponding letters, $p_x$ excluded,
\ie \[ P_x\triangleq \{p_z\in P \mid z \text{ appears in } \pi_x \text{ and } p_z \text{ is the letter associated with } z \text{ in } \pi_T\}\setminus\{p_x\}.\]
Let $s_x$ be the partial scenario defined as follows:
\[
	s_x : P_x \rightarrow \{0,1\} : \left\{\begin{array}{ll}
					\lambda\;, & \text{ if } x=r; \\
					p_z\mapsto b_{(z, z')}, & \text{ if } x\neq r.
				  \end{array}\right.
\]
where $z'$ is the unique child of $z$ in $\pi_T$ lying on $\pi_x$.
Let $x_0$ (resp., $x_1$) be the unique child of $x$ in $\pi_T$
	such that $b_{x,x_0}=0$ (resp., $b_{x,x_1}=1$).
For every complete $s'_x\in\Sigma_P$ such that $\Sub(s'_x,s_x)$, we define the following set of arcs:
\[ B'_{s'_x}\triangleq
\left\{	\begin{array}{ll}
		\big\{\big( (\Ord_{p_{x_0}})_{s'_x} , (\Ord_{p_x})_{s'_x}, 0 \big)\big\},  & \text{ if } s'_x(x)=0;\\
		\big\{\big( (\Ord_{p_{x_1}})_{s'_x}, (\Ord_{p_x})_{s'_x}, 0 \big)\big\},  & \text{ if } s'_x(x)=1.
	\end{array}\right.
\]
Also, for every complete $s'_x,s''_x\in\Sigma_P$ such that $\Sub(s'_x,s_x)$ and $\Sub(s''_x,s_x)$, where $s'_x\neq s''_x$, we define:
\[ C'_{s'_x, s''_x}\triangleq \big\{\big( (\Ord_{p_x})_{s'_x} , (\Ord_{p_x})_{s''_x}, 0 \big)\big\}.\]

Finally,
\[
	A'_x \triangleq \bigcup_{\footnotesize s'_x\in\Sigma_P \text{ : } \normalsize \Sub(s'_x, s_x)} B'_{s'_x}
		\cup \bigcup_{\shortstack{\footnotesize$s'_x, s''_x \in\Sigma_P$ \text{:}
			\footnotesize $s'_x\neq s''_x$, \\
			\footnotesize $\Sub(s'_x,s_x), \Sub(s''_x,s_x) $ }} C'_{s'_x, s''_x}.
\]

\item Then, $\H^\pi_{0}(\Gamma)$ is defined as
	$\H^\pi_{0}(\Gamma)\triangleq ( V^{\text{Ex}}_{\Gamma}, \A_{\H^\pi_{0}(\Gamma)})$,
	where, \[\A_{\H^\pi_{0}(\Gamma)}\triangleq \Lambda^{\text{Ex}}_{\Gamma} \cup
		\bigcup_{\substack{s_1,s_2\in\Sigma_P \\ u\in V^+_{s_1,s_2}}}
			\alpha_{\varepsilon}(s_1;s_2;u) \cup
		\bigcup_{\shortstack{$x$ \footnotesize : internal \\ \footnotesize node of \normalsize $\pi_T$}} A'_x .\]
\end{itemize}
\end{Def}

Notice that the following holds:
each $\alpha_{\varepsilon}(s_1; s_2; u)$ has size:
 \[|\alpha_{\varepsilon}(s_1; s_2; u)| = |\Delta(s_1;s_2)|+1\leq |P|+1.\]

The following theorem establishes the connection between the \oDCC of \CSTN{s} and the consistency of \HTN{s}.

\begin{Thm}\label{thm:mainreduction}
Given any \CSTN $\Gamma=( V, A, L, \Ord, {\Ord}V, P )$, it holds that the \CSTN
$\Gamma$ is \oDCC if and only if there exists a c-ps-tree $\pi_T$
such that the \HTN $\H^{\pi_T}_{0}(\Gamma)$ is consistent.

Moreover, $\H^{\pi_T}_{0}(\Gamma)$ has at most so many nodes: \[|V_{\H^{\pi_T}_{0}(\Gamma)}|\leq |\Sigma_P|\cdot |V|,\]
so many hyperarcs: \[|\A_{\H^{\pi_T}_{0}(\Gamma)}|=O(|\Sigma_P|\cdot |A| + |\Sigma_P|^2|V|),\]
and it has size at most: \[ m_{\A_{\H^{\pi_T}_{0}(\Gamma)}}=O(|\Sigma_P|\cdot |A| + |\Sigma_P|^2|V|\cdot |P|).\]
\end{Thm}
\begin{proof}
(1) Firstly, we prove that the \CSTN $\Gamma$ is \oDCC if and only if there exists
a c-ps-tree $\pi_T$ such that the \HTN $\H^{\pi_T}_{0}(\Gamma)$ is consistent.

($\Rightarrow$) Let $\sigma\in\S_{\Gamma}$ be a given viable and $\pi$-dynamic execution strategy for the \CSTN $\Gamma$.
Since $\sigma$ is $\pi$-dynamic, then for any two $s_1, s_2\in \Sigma_P$ and any $v\in V^+_{s_1,s_2}$
the following holds on $\tau\triangleq [\sigma(s_1)]^t_v$ and $\psi\triangleq [\sigma(s_1)]^{\pi}_v$:
\[ \Con(\oHst(\tau, \psi, s_1, \sigma), s_2) \Rightarrow [\sigma(s_2)]^t_v = \tau, [\sigma(s_2)]^{\pi}_v = \psi.\]
It is easy to see that this induces one and only one c-ps-tree $\pi_T$: indeed,
due to Remark~\ref{rem:tree_structure},
there must be exactly one $\Ord_{p'}\in\Ord{V}$, for some $p'\in P$, which is executed at first
(\wrt to both execution time and position) under all possible scenarios; then, depending on the boolean result of $p'$,
a second observation $p''$ can be differentiated, and it can occur at the same or at a subsequent time instant,
but still at a subsequent position; again, by Remark~\ref{rem:tree_structure},
there is exactly one $\Ord_{p''}\in\Ord{V}$ which comes first under all possible scenarios that agree on $p'$;
and so on and so forth, thus forming an arborescence over $P$, rooted at $p'$, which is captured uniquely by a c-ps-tree.
Then, let $\phi_{\sigma}:V^{\text{Ex}}_{\Gamma}\rightarrow\RR$ be the schedule of $\H^{\pi_T}_{0}(\Gamma)$ defined as:
$\phi_{\sigma}(v_s)\triangleq [\sigma(s)]^t_v$ for every $v_s\in V^{\text{Ex}}_{\Gamma}$, where $s\in\Sigma_P$ and $v\in V^+_s$.
It is not difficult to check from the definitions, at this point, that all of the standard arc
and hyperarc constraints of $\H^{\pi_T}_{0}(\Gamma)$ are satisfied by $\phi_{\sigma}$,
that is to say that $\phi_{\sigma}$ must be feasible for $\H^{\pi_T}_{0}(\Gamma)$.
Hence, $\H^{\pi_T}_{0}(\Gamma)$ is consistent.

($\Leftarrow$)
Assume that there exists a c-ps-tree $\pi_T$ such that the \HTN $\H^{\pi_T}_{0}(\Gamma)$ is consistent,
and let $\phi:V_\Gamma^{\text{Ex}}\rightarrow\RR$ be a feasible schedule for $\H^{\pi_T}_{0}(\Gamma)$.
Then, let $\sigma_{\phi, \pi_T}(s)\in\S_{\Gamma}$ be the execution strategy defined as follows:
\begin{itemize}
\item $[\sigma_{\phi, \pi_T}(s)]^t_v\triangleq \phi(v_s)$,
	$\forall$ $v_s\in V_{\Gamma}^{\text{Ex}}$, $s\in\Sigma_P$, $v\in V^+_s$;
\item Let $s'\in\Sigma_P$ be any complete scenario.
Then, $s'$ induces exactly one path in $\pi_T$, in a natural way,
\ie by going from the root $r$ down to some leaf $\gamma$.
Notice that the sequence of labels $(p_r, p_{v_2}, \ldots, p_{\gamma})$ can be seen as a bijection,
\ie $\pi_\gamma : \Ord{V}^+_{s'} \rightleftharpoons \{1, \ldots, |\Ord{V}^+_{s'}|\}$.
Then, for any $s'\in \Sigma_P$ and $v\in \Ord{V}^+_{s'}$,
let $[\sigma_{\phi}^{\pi_T}(s')]^\pi_v\triangleq \pi_\gamma(v)$.
\end{itemize}
It is not difficult to check from the definitions, at this point, that since $\phi$ is feasible for $\H^{\pi_T}_{0}(\Gamma)$,
then $\sigma^{\pi_T}_{\phi}$ must be viable and $\pi$-dynamic for the \CSTN $\Gamma$. Hence, the \CSTN $\Gamma$ is \oDCC.

(2) The size bounds for $\H^{\pi_T}_{0}(\Gamma)$ follow from Definition~\ref{def:Hepsilonzero}.
\end{proof}

Algorithm~\ref{FIG:PseudocodeReduction-cstn-htn}
contains the pseudocode for constructing the \HTN $\H^{\pi_T}_{0}(\Gamma)$,
as prescribed by Definition~\ref{def:Hepsilonzero}.
\begin{algorithmenv}[!htb]
\removelatexerror
\begin{algorithm}[H]\label{algo:pseudocode_construct_H}
\caption{$\texttt{construct\_}\H$$(\Gamma, \pi_T)$}
\KwIn{a \CSTN $\Gamma\triangleq( V, A, L, \Ord, {\Ord}V, P )$, a c-ps-tree $\pi_T$ coherent with $\Gamma$.}
\ForEach{($s\in\Sigma_P$)}{
$V_s\leftarrow\{v_s \mid v\in V^+_{s}\}$\;
$A_s\leftarrow\{a_s\mid a\in A^+_s\}$\;

}
$\displaystyle V^{\text{Ex}}_{\Gamma}\leftarrow \cup_{\substack{s\in\Sigma_P}}V_s$\;
$\displaystyle \Lambda^{\text{Ex}}_{\Gamma}\leftarrow \cup_{\substack{s\in\Sigma_P}} A_s$\;
\ForEach{($s_1, s_2\in\Sigma_P$, $s_1\neq s_2$)}{
	\ForEach{($u\in V^+_{s_1, s_2}\setminus \Ord{V}$)}{
		$t_\alpha\leftarrow u_{s_1}$\;
		$H_\alpha\leftarrow \{u_{s_2}\}\cup (\Delta(s_1; s_2))$\;
		$w_\alpha(u_{s_2})\leftarrow 0$\;
		\ForEach{$v\in\Delta(s_1; s_2)$}{
			$w_\alpha(v_{s_1})\leftarrow 0$\;
		}
		$\alpha_{0}(s_1; s_2; u)\leftarrow ( t_\alpha, H_\alpha, w_\alpha)$\;
	}
}
\ForEach{($x$ : internal node of $\pi_T$)}{
$A'_x\leftarrow $ as defined in Definition~\ref{def:Hepsilonzero};
}
$\displaystyle\A_{\H^{\pi_T}_{0}(\Gamma)}\leftarrow \Lambda^{\text{Ex}}_{\Gamma} \cup
\bigcup_{\substack{s_1,s_2\in\Sigma_P \\ u\in V^+_{s_1, s_2}}}\alpha_{0}(s_1;s_2;u)
	\cup  \bigcup_{\shortstack{$x$ : \footnotesize internal \\ \footnotesize node of \normalsize $\pi_T$}} A'_x $\;
$\H^{\pi_T}_{0}(\Gamma)\leftarrow ( V^{\text{Ex}}_{\Gamma}, \A_{H^{\pi_T}_{0}(\Gamma)})$\;
\Return{$\H^{\pi_T}_{0}(\Gamma)$;}
\end{algorithm}
\caption{Constructing the \HTN $\H^{\pi_T}_{0}(\Gamma)$.}
\end{algorithmenv}
If $\Gamma$ is \oDCC, there is an integer weighted $\pi$-dynamic $\pi$-ES, as below.
\begin{Prop}\label{prop:integral}
Assume $\Gamma=( V, A, L, \Ord, {\Ord}V, P )$ to be \oDCC.
Then, there is some $\pi$-ES $\sigma\in \S_\Gamma$ which is viable, $\pi$-dynamic, and \emph{integer weighted},
namely, for every $s\in\Sigma_{P}$ and every $v\in V^+_s$, the following property holds:
\[ [\sigma(s)]^t_v\in \big\{0, 1, 2, \ldots, \mathcal{M}_\Gamma\big\}\subseteq \N, \]
where $\mathcal{M}_\Gamma\triangleq \big(|\Sigma_P||V| + |\Sigma_P||A| + |\Sigma_P|^2|V|\big)W$.
\end{Prop}
\begin{proof} By Theorem~\ref{thm:mainreduction}, since $\Gamma$ is \oDCC,
there exists some c-ps-tree $\pi_T$ such that the \HTN $\H^{\pi_T}_{0}(\Gamma)$ is consistent;
moreover, by Theorem~\ref{thm:mainreduction} again, $\H^{\pi_T}_{0}(\Gamma)$
has $|V_{\H^{\pi_T}_{0}(\Gamma)}|\leq |\Sigma_P|\, |V|$ nodes and
$|\A_{\H^{\pi_T}_{0}(\Gamma)}|\leq |\Sigma_P|\,|A| + |\Sigma_P|^2|V|$ hyperarcs.
Since $\H^{\pi_T}_{0}(\Gamma)$ is consistent, it follows from
Theorem~\ref{thm:mainreduction} (also see Lemma~1 and Theorem~8 in \citet{CPR2016})
that $\H^{\pi_T}_{0}(\Gamma)$ admits an integer weighted and feasible schedule $\phi$ such that:
\[\phi: V_{\H^{\pi_T}_{0}(\Gamma)}\rightarrow \big\{0, 1, 2, \ldots, \mathcal{M}_\Gamma \big\},\]
where $\mathcal{M}_\Gamma\leq (|V_{\H^{\pi_T}_{0}(\Gamma)}|+|\A_{\H^{\pi_T}_{0}(\Gamma)}|)W$.

Therefore, it holds that $\mathcal{M}_\Gamma\leq (|\Sigma_P|\, |V| + |\Sigma_P|\,|A| + |\Sigma_P|^2|V|)W$.
\end{proof}

Given a \CSTN $\Gamma$ and some c-ps-tree $\pi_T$, it is thus easy to check whether there exists some $\pi$-ES for $\Gamma$
whose ordering relations are exactly the same as those prescribed by $\pi_T$. Indeed,
it is sufficient to construct $\H^{\pi_T}_{0}(\Gamma)$ with Algorithm~\ref{algo:pseudocode_construct_H},
then checking the consistency of $\H^{\pi_T}_{0}(\Gamma)$ with the algorithm mentioned in Theorem~\ref{Teo:MainAlgorithms}.
This results into Algorithm~\ref{algo:solve_piDCC}.
The corresponding time complexity is also that of Theorem~\ref{Teo:MainAlgorithms}.
\begin{algorithmenv}[!htb]
\removelatexerror
\begin{algorithm}[H]\label{algo:solve_piDCC}
\caption{\texttt{check\_\oDCC\_on\_c-ps-tree}$(\Gamma, \pi_T)$}
\KwIn{a \CSTN $\Gamma\triangleq( V, A, L, \Ord, {\Ord}V, P )$,
	a c-ps-tree $\pi_T$ coherent with $\Gamma$.}
$\H^{\pi_T}_{0}(\Gamma)\leftarrow\texttt{construct\_}\H(\Gamma, \pi_T)$;
			\tcp{\footnotesize ref. Algorithm~\ref{algo:pseudocode_construct_H}}
$\phi\leftarrow \texttt{check\_\HTNC}(\H^{\pi_T}_{0}(\Gamma))$; \tcp{\footnotesize ref. Thm~\ref{Teo:MainAlgorithms}}
\If{($\phi$ is a feasible schedule of $\H^{\pi_T}_{0}(\Gamma)$)}{
	\Return{$(\texttt{YES}, \phi, \pi_T)$;}
}
\Return{$\texttt{NO}$;}
\end{algorithm}
\caption{Checking \oDCC given a c-ps-tree $\pi_T$, by reduction to \HTNC.}
\label{FIG:PseudocodeReduction-cstn-htn}
\end{algorithmenv}

Notice that, in principle, one could generate all of the possible c-ps-trees $\pi_T$ given $P$, one by one, meanwhile checking
for the consistency state of $\H^{\pi_T}_{0}(\Gamma)$ with Algorithm~\ref{algo:solve_piDCC}.
However, it is not difficult to see that, in general,
the total number $f_{|P|}$ of possible c-ps-trees over $P$ is not singly-exponential in $|P|$.
Indeed, a moment's reflection revelas that for every $n>1$ it holds that $f_n=n\cdot f_{n-1}^2$,
and $f_1=1$. So, any algorithm based on the exhaustive exploration of the whole space comprising all of
the possible c-ps-trees over $P$ would not have a (pseudo) singly-exponential time complexity in $|P|$.
Nevertheless, we have identified another solution, that allows us to provide a sound-and-complete (pseudo)
singly-exponential time \oDCC-Checking procedure: it is a simple and self-contained
reduction from \oDCC-Checking to \DCC-Checking. This allows us to provide the first
sound-and-complete (pseudo) singly-exponential time \oDCC-Checking algorithm which employs our previous
\DCC-Checking algorithm (\ie that underlying Theorem~\ref{thm:mainresult}) in a direct manner,
as a black box, thus avoiding a more fundamental restructuring of it.

\subsection{A Singly-Exponential Time \oDCC-Checking Algorithm}\label{subsect:Algo}
This section presents a sound-and-complete (pseudo) singly-exponential time algorithm for solving \oDCC,
also producing a viable and $\pi$-dynamic $\pi$-ES whenever the input \CSTN is really \oDCC.

The main result of this section goes as follows.
\begin{Thm}\label{thm:mainresult_pi-DC}
There exists an algorithm for checking \oDCC on any input given \CSTN $\Gamma = (V, A, L, \Ord, \Ord{V}, P)$
with the following (pseudo) singly-exponential time complexity:
\[ O\Big(|\Sigma_P|^{3}|A|^2|V|^2 +|\Sigma_P|^4|A||V|^3|P| + |\Sigma_P|^5|V|^4|P|\Big)W. \]
Moreover, when $\Gamma$ is \oDCC, the algorithm also returns a viable and $\pi$-dynamic $\pi$-ES for $\Gamma$.
Here, $W\triangleq \max_{a\in A} |w_a|$.
\end{Thm}

The algorithm mentioned in Theorem~\ref{thm:mainresult_pi-DC} consits of a simple
reduction from \oDCC to $\epsilon$-DC for $\epsilon=1$, \ie $1$-DC.

Firstly, we illustrate the technique on which the algorithm is based on by showing a reduction from
\oDCC to (classical) DC; this same reduction already appeared in~\citet{CCR16},
	but here the presentation is improved.

Secondly, we will show in Subsection~\emph{``Improving Algorithm~\ref{algo:check_pi-DC}"} how to leverage the technique in order to reduce \oDCC to $1$-DC,
thus obtaining the time complexity bound mentioned in Theorem~\ref{thm:mainresult_pi-DC}; this is a novel contribution \wrt~\citet{CCR16}.

Basically, the idea underlying the proof technique is to give a (sufficiently small, rational) margin $\gamma$ so that the planner can actually do before,
in the sense of the time value $[\sigma(s)]_v$, what he did ``before" in the ordering $\pi$.
Given any ES in the relaxed network, the planner would then turn it into a $\pi$-ES for the original network
(which has some more stringent constraints), by rounding-down each time value $[\sigma(s)]_v$ to the largest integer
less than or equal to it, \ie $\big\lfloor [\sigma(s)]_v \big\rfloor$.
The problem is that one may (possibly) violate some constraints when there
is a ``leap" in the rounding (\ie a difference of one unit, in the rounded value,
\wrt what one would have wanted). Anyhow, we have identified a technique
that allows us to get around this subtle case, provided that $\gamma$ is exponentially small.

\begin{Def}{Relaxed \CSTN $\Gamma'$.}\label{def:relaxed_cstn}
Let $\Gamma=( V, A, L, \Ord, \Ord{V}, P )$ be any \CSTN with integer constraints. Let $\gamma \in (0,1)$ be a rational.
Define $\Gamma'_\gamma \triangleq ( V, A'_\gamma, L, \Ord, \Ord{V}, P )$
to be a \CSTN that differs from $\Gamma$ only in the numbers
appearing in the constraints. Specifically, each constraint $( u-v \leq \delta, \ell )\in A$
is replaced in $\Gamma'_\gamma$ by a slightly relaxed constraint, $( u-v\leq \delta'_\gamma, \ell ) \in A'_\gamma$, where:
\[\delta'_\gamma\triangleq \delta+|V|\cdot\gamma. 
\]
\end{Def}
The following two lemmata hold for any \CSTN $\Gamma$.

\begin{Lem}\label{lem:reduction_easy}
Let $\gamma$ be any rational in $(0, |V|^{-1})$. If $\Gamma$ is \oDCC, then $\Gamma'_\gamma$ is DC.
\end{Lem}
\begin{proof}
Since $\Gamma$ is \oDCC, by Proposition~\ref{prop:integral}, there exists an integer weighted,
viable and $\pi$-dynamic, $\pi$-ES $\sigma$ for $\Gamma$.
Let us fix some rational $\gamma\in (0,|V|^{-1})$.
Define the ES $\sigma'_\gamma\in\S_{\Gamma'_\gamma}$ as follows, for every $s\in\Sigma_P$ and $v\in V^+_s$:
\[ [\sigma'_\gamma(s)]_v \triangleq [\sigma(s)]^t_v +[\sigma(s)]^\pi_v\cdot \gamma.\]
Since $[\sigma(s)]^\pi_v\leq |V|$, then: \[ [\sigma(s)]^\pi_v\cdot \gamma < |V|\cdot|V|^{-1} = 1,\]
and so the total ordering of the values $[\sigma'_\gamma(s)]_v$, for a given $s\in \Sigma_P$, coincides with $[\sigma(s)]^\pi$.
Hence, the fact that $\sigma'_\gamma$ is dynamic follows directly from the $\pi$-dynamicity of $\sigma$.
Moreover, no LTC $(u-v\leq \delta'_\gamma, \ell)$ of $\Gamma'_\gamma$ is violated in any scenario $s\in\Sigma_{P}$ since,
if $\Delta'_{\gamma, u,v} \triangleq [\sigma'_\gamma(s)]_u - [\sigma'_\gamma(s)]_v$ then:
\begin{align*}
\Delta'_{\gamma, u,v} & = \big([\sigma(s)]^t_u + [\sigma(s)]^\pi_u\cdot\gamma \big)
								- \big([\sigma(s)]^t_v + [\sigma(s)]^\pi_v \cdot\gamma \big)  \\
				& \leq [\sigma(s)]^t_u - [\sigma(s)]^t_v + |V|\cdot \gamma     \\
				& \leq \delta + |V|\cdot \gamma = \delta'.
\end{align*}
So, $\sigma'_\gamma$ is viable. Since $\sigma'_\gamma$ is also dynamic, then $\Gamma'_\gamma$ is DC.
\end{proof}

The next lemma shows that the converse direction holds as well, but for (exponentially) smaller values of $\gamma$.
\begin{Lem}\label{lem:reduction_hard} Let $\gamma$ be any rational in $(0, |\Sigma_P|^{-1}\cdot |V|^{-2})$.
If $\Gamma'_\gamma$ is DC, then $\Gamma$ is \oDCC. \end{Lem}
\begin{proof}
Let $\sigma'_\gamma\in\S_{\Gamma'_\gamma}$ be some viable and dynamic ES for $\Gamma'_\gamma$.

Firstly, we aim at showing that, w.l.o.g., the following lower bound holds:
\[ [\sigma'_\gamma(s)]_v - \big\lfloor [\sigma'_\gamma(s)]_v \big\rfloor \geq |V|\cdot \gamma,
\text{ for \emph{all} } s\in \Sigma_{P} \text{ and } v\in V^+_s. \hspace*{.75cm} \text{(LB)} \]

This will allow us to simplify the rest of the proof.
In order to prove it, let us pick any $\eta\in [0,1)$ such that:
\[ [\sigma'_\gamma(s)]_v-\eta-k \in [0, |V|\cdot \gamma),\; \text{ for \emph{no} } v\in V,\; s\in\Sigma_P,\; k\in \Z. \]
Observe that such a value $\eta$ exists. Indeed, there are only $|\Sigma_P|\cdot |V|$
choices of pairs $(s,v)\in\Sigma_P\times V$ and each pair rules out a (circular)
semi-open interval of length $|V|\cdot \gamma$ in $[0,1)$, so the total measure of invalid values for $\eta$ in the
semi-open real interval $[0,1)$ is at most $|\Sigma_P|\cdot |V|\cdot |V|\cdot \gamma < 1$. So $\eta$ exists.

See \figref{fig:reduction_illustration} for an intuitive illustration of this fact,
where three semi-open intervals are depicted in orange, green, and yellow color (respectively);
note that the green interval and the yellow one partially overlap at the end (beginning) of the former (latter),
and that the yellow interval includes both extremes of the unit interval $[0,1)$ (\ie it is circular).

\begin{figure}[h!]
\begin{center}
\begin{tikzpicture}[scale=1.2, xscale=1.2]
\node[scale=.3, label={below left, xshift=1ex, yshift=-1ex:$[0$} ] (A) at (0,.1) {};
\node[scale=.3] (B) at (1,.1) {};
\node[scale=.3] (C) at (2,.1) {};
\node[scale=.3] (D) at (3,.1) {};
\node[scale=.3] (E) at (4,.1) {};
\node[scale=.3] (F) at (5,.1) {};
\node[-, label={below right, xshift=6ex, xshift=-1ex, yshift=-.6ex:$1)$}] (G) at (6,.1) {};

\draw[very thick, draw=black!35!blue](3.75,-.225)--(3.5,-.225)--(3.5,.225)--(3.75,.225);
	\node[node, fill=blue, scale=.4] (ETA) at (3.5,0) {};
\draw[very thick, draw=black!35!blue, scale=.2](19.925,0) arc (0:45:1.8);
\draw[very thick, draw=black!35!blue, scale=.2](19.925,0) arc (0:-45:1.8);

\shade[top color=yellow, bottom color=black!50] (0,.15) rectangle (.5,-.15);

\draw[very thick, draw=black!35!orange](1.25,-.225)--(1,-.225)--(1,.225)--(1.25,.225);
	\shade[top color=orange, bottom color=black!50](1,.15) rectangle (3,-.15);
\draw[very thick, draw=black!35!orange, scale=.2](15.1,0) arc (0:45:1.8);
\draw[very thick, draw=black!35!orange, scale=.2](15.1,0) arc (0:-45:1.8);

\draw[very thick, draw=black!35!green](4.25,-.225)--(4,-.225)--(4,.225)--(4.25,.225);
	\shade[top color=green, bottom color=black!50] (4,.15) rectangle (5.5,-.15);

\shade[top color=yellow, bottom color=black!40!green] (5.5,.15) rectangle (6,-.15);
\shade[top color=yellow, bottom color=black!50] (6,.15) rectangle (7,-.15);

\draw[very thick, draw=black!35!green, scale=.2](30.1,0) arc (0:45:1.8);
\draw[very thick, draw=black!35!green, scale=.2](30.1,0) arc (0:-45:1.8);

\draw[very thick, draw=black!35!yellow](5.75,-.225)--(5.5,-.225)--(5.5,.225)--(5.75,.225);

\draw[very thick, draw=black!35!yellow, scale=.2](2.6,0) arc (0:45:1.8);
\draw[very thick, draw=black!35!yellow, scale=.2](2.6,0) arc (0:-45:1.8);

\draw[-][thick] (0,0) -- (4,0);
\draw[-][thick] (4,0) -- (5,0);
\draw[-][thick]  (5,0) -- (7,0);

\draw [thick] (0,-.1) node[below]{} -- (0,0.1);
\draw [thick] (.5,-.1) node[below, yshift=-.25ex]{\footnotesize$\gamma$} -- (.5,0.1);
\draw [thick] (1,-.1) node[below, yshift=-.25ex]{\footnotesize$2\gamma$} -- (1,0.1);
\draw [thick] (1.5,-.1) node[below, yshift=-.25ex]{\footnotesize$3\gamma$} -- (1.5,0.1);
\draw [thick] (2,-.1) node[below]{} -- (2,0.1);
\draw [thick] (2.5,-.1) node[below, yshift=-.25ex]{$\cdots$} -- (2.5,0.1);
\draw [thick] (3,-.1) node[below, xshift=2.5ex, yshift=-.25ex]{\footnotesize $j\gamma$} -- (3,0.1);
\draw [thick] (3.5,-.1) node[above, xshift=-1ex, yshift=1ex]{\large$\eta$} -- (3.5,0.1);
\draw [thick] (4,-.1) node[below, xshift=1.75ex, yshift=-.25ex]{\footnotesize$(j+1)\gamma$} -- (4,0.1);
\draw [thick] (4.5,-.1) node[below]{} -- (4.5,0.1);
\draw [thick] (5,-.1) node[below]{} -- (5,0.1);
\draw [thick] (5.5,-.1) node[below, yshift=-.25ex]{$\cdots$} -- (5.5,0.1);
\draw [thick] (6,-.1) node[below]{} -- (6,0.1);
\draw [thick] (6.5,-.1) node[below, yshift=-.25ex]{\footnotesize$1-\gamma$} -- (6.5,0.1);
\draw [thick] (7,-.1) node[below]{} -- (7,0.1);

\draw[thick, decoration={brace,raise=5pt},decorate]  (4,.2) -- node[above=2ex] {$|V|\cdot\gamma$} (6,.2);
\end{tikzpicture}
\end{center}
\caption{An illustration of the proof of Lemma~\ref{lem:reduction_hard}.}\label{fig:reduction_illustration}
\end{figure}
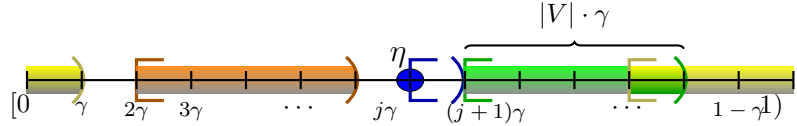

By adding some translation parameter $\eta$ to all time values $\{[\sigma'_\gamma(s)]_v\}_{v\in V, s\in\Sigma_P}$
	we can assume w.l.o.g. that $\eta=0$ holds for the rest of this proof; and thus, that (LB) holds.
Now, define $[\sigma(s)]^t_v \triangleq \big\lfloor [\sigma'_\gamma(s)]_v \big\rfloor$,
and let $[\sigma(s)]^\pi$ be the ordering induced by $\sigma'_\gamma(s)$.
Observe that $\sigma$ is a well-defined $\pi$-ES (\ie that $[\sigma(s)]^\pi$ is coherent \wrt $[\sigma(s)]^t$),
thanks to the fact that $\lfloor \cdot \rfloor$ is a monotone operator.
Since the ordering $[\sigma(s)]^\pi$ is the same as that of $\sigma'_\gamma(s)$, then $\sigma$ is $\pi$-dynamic.

It remains to prove that $\sigma$ is viable. For this, take any constraint $(u-v\leq \delta, \ell)\in A$ in $\Gamma$,
and suppose that: \[ [\sigma'_\gamma(s)]_u-[\sigma'_\gamma(s)]_v\leq \delta'_\gamma
	= \delta + |V|\cdot\gamma. \hspace*{3ex} \text{(A)}\]
If $[\sigma'_\gamma(s)]_u-[\sigma'_\gamma(s)]_v\leq \delta$,
and $\delta\in \Z$,	then clearly $[\sigma(s)]^t_u-[\sigma(s)]^t_v\leq \delta$.
So, the interesting case is when:
	\[ 0 < [\sigma'_\gamma(s)]_u - [\sigma'_\gamma(s)]_v-\delta\leq |V|\cdot\gamma.\]
For this, we observe that the following ($\star$) holds by (LB):
	\[ \big\lfloor [\sigma'_\gamma(s)]_u \big\rfloor \leq [\sigma'_\gamma(s)]_u -|V|\cdot\gamma. \hspace*{2ex} \text{($\star$)}\]
Also, it is clear that:
	\[ \big\lfloor [\sigma'_\gamma(s)]_v \big\rfloor > [\sigma'_\gamma(s)]_v - 1. \hspace*{8ex} \text{($\star\star$)}\]
Then,
\begin{align*}\hspace*{-2ex}
[\sigma(s)]^t_u - [\sigma(s)]^t_v & =
				\big\lfloor [\sigma'_\gamma(s)]_u \big\rfloor - \big\lfloor [\sigma'_\gamma(s)]_v \big\rfloor  & \text{(by def. of $[\sigma(s)]^t_x$,  $x\in\{u,v\}$)}\\
				  & < ([\sigma'_\gamma(s)]_u-|V|\cdot\gamma) - ([\sigma'_\gamma(s)]_v - 1) & \text{(by ($\star$) and ($\star\star$))}\\
					  & = ([\sigma'_\gamma(s)]_u - [\sigma'_\gamma(s)]_v) - |V|\cdot\gamma + 1 & \text{(by rewriting)} \\
						& \leq \delta'_\gamma-|V|\cdot\gamma+1 & \text{(by (A))} \\
				  & = \delta + 1. & \text{(by $\delta'_\gamma=\delta+|V|\cdot\gamma$)}
\end{align*}
Now, since we have the strict inequality $[\sigma(s)]^t_u - [\sigma(s)]^t_v < \delta + 1$,
and since $[\sigma(s)]^t_u - [\sigma(s)]^t_v\in\Z$, then $[\sigma(s)]^t_u - [\sigma(s)]^t_v \leq \delta$, as desired.
So, $\sigma$ is viable. Since $\sigma$ is both viable and $\pi$-dynamic, then $\Gamma$ is \oDCC.
\end{proof}

\figref{fig:reduction_illustration} illustrates the proof of Lemma~\ref{lem:reduction_hard},
in which a family of (circular) semi-open intervals of length $|V|\cdot\gamma$ are depicted as shaded rectangles.
Lemma~\ref{lem:reduction_hard} ensures that at least one chunk of length $l_\gamma \geq 1-|\Sigma_P|\cdot|V|^2\cdot\gamma$
is not covered by the union of those (circular) semi-open intervals, and it is therefore free to host $\eta$;
in \figref{fig:reduction_illustration}, this is represented by the blue interval,
and $\eta=j\cdot \gamma$ for some $j\in [0, \gamma^{-1})$. Also notice that $\gamma$ can be fixed as follows:
\[ \gamma \triangleq \frac{1}{|\Sigma_P|\cdot|V|^2+1}; \]
then, $l_\gamma \geq 1-|\Sigma_P|\cdot|V|^2\cdot\gamma = \gamma$.

In summary, Lemma~\ref{lem:reduction_easy} and Lemma~\ref{lem:reduction_hard} imply Theorem~\ref{thm:reduction}.
\begin{Thm}\label{thm:reduction}
Let $\Gamma$ be a \CSTN and let $\gamma\in (0, |\Sigma_P|^{-1}\cdot|V|^{-2})$.
Then, $\Gamma$ is \oDCC if and only if $\Gamma'_\gamma$ is DC.
\end{Thm}
This allows us to design a simple algorithm for solving \oDCC-Checking, by reduction to \DCC-Checking,
which is named $\texttt{Check-\oDCC}()$ (Algorithm~\ref{algo:check_pi-DC}). Its pseudocode follows below.

\begin{algorithmenv}[!h]
\removelatexerror
\begin{algorithm}[H]\label{algo:check_pi-DC}
\caption{\texttt{Check-\oDCC}$(\Gamma)$}
\KwIn{a \CSTN $\Gamma\triangleq ( V, A, L, \Ord, {\Ord}V, P )$}
$\gamma\leftarrow \frac{1}{|\Sigma_P|\cdot |V|^2+1}$\label{algo:check_pi-DC:l1}\;
$A'_\gamma\leftarrow \big\{( u-v\leq \delta+|V|\cdot\gamma, \ell )
	\mid ( u-v\leq \delta, \ell )\in A \big\}$\label{algo:check_pi-DC:l2}\;
$\Gamma'_\gamma\leftarrow ( V, A'_\gamma, L, \Ord, {\Ord}V, P )$\label{algo:check_pi-DC:l3}\;
$\sigma'_\gamma\leftarrow \texttt{check\_\DCC}(\Gamma'_\gamma)$;\label{algo:check_pi-DC:l4}
	\tcp{\footnotesize see Theorem~\ref{thm:mainresult}}
\If{$\sigma'_\gamma$ is a viable and dynamic ES for $\Gamma'_\gamma$\label{algo:check_pi-DC:l5}}{
$\eta\leftarrow $ pick $\eta\in [0,1)$ as in the proof of Lemma~\ref{lem:reduction_hard}\label{algo:check_pi-DC:l6}\;
\ForEach{$(s,v)\in \Sigma_P\times V^+_s$\label{algo:check_pi-DC:l7}}{
   $[\sigma'_\gamma(s)]_v \leftarrow [\sigma'_\gamma(s)]_v - \eta$;\label{algo:check_pi-DC:l8} \tcp{\footnotesize shift by $\eta$;}
}
let $\sigma\in\Sigma_\Gamma$ be constructed as follows\label{algo:check_pi-DC:l9}\;
\ForEach{$s\in \Sigma_P$\label{algo:check_pi-DC:l10}}{
	\ForEach{$v\in V^+_s$\label{algo:check_pi-DC:l11}}{
		$[\sigma(s)]^t_v\leftarrow \Big\lfloor [\sigma'_\gamma(s)]_v \Big\rfloor$\label{algo:check_pi-DC:l12}\;
	}
	$[\sigma(s)]^\pi\leftarrow$ the ordering on $P$ induced by $\sigma'_\gamma(s)$\label{algo:check_pi-DC:l13}\;
}
\Return{$(\texttt{YES}, \sigma)$;}\label{algo:check_pi-DC:l14}
}
\Return{\texttt{NO}\label{algo:check_pi-DC:l15}}\;
\end{algorithm}
\caption{Checking \oDCC by reduction to \DCC-Checking.}
\end{algorithmenv}

\paragraph*{Description of Algorithm~\ref{algo:check_pi-DC}}
It takes in input a \CSTN~$\Gamma$. When $\Gamma$ is \oDCC, it aims at returning $(\texttt{YES}, \sigma)$,
where $\sigma\in\S_\Gamma$ is a viable and $\pi$-dynamic $\pi$-ES for $\Gamma$.
Otherwise, if $\Gamma$ is not \oDCC, then $\texttt{Check-\oDCC}()$ (Algorithm~\ref{algo:check_pi-DC}) returns \texttt{NO}.
Of course the algorithm implements the reduction described in Definition~\ref{def:relaxed_cstn},
whereas the $\pi$-ES is computed as prescribed by Lemma~\ref{lem:reduction_hard}.
At line~\ref{algo:check_pi-DC:l1}, we set $\gamma\leftarrow \frac{1}{|\Sigma_P|\cdot |V|^2+1}$.
Then, at lines~\ref{algo:check_pi-DC:l2}-\ref{algo:check_pi-DC:l3},
$\Gamma'_\gamma$ is constructed as in Definition~\ref{def:relaxed_cstn},
\ie $\Gamma'_\gamma\leftarrow ( V, A'_\gamma, L, \Ord, {\Ord}V, P )$,
where $A'_\gamma\leftarrow
\big\{( u-v\leq \delta+|V|\cdot\gamma, \ell ) \mid ( u-v\leq \delta, \ell )\in A \big\}$.
At this point, at line~\ref{algo:check_pi-DC:l5}, the
\DCC-Checking algorithm of Theorem~\ref{thm:mainresult} is invoked on input $\Gamma'_\gamma$.
Let $\sigma'_\gamma$ be its output. If $\Gamma'_\gamma$ is not DC,
then $\texttt{Check-\oDCC}()$ (Algorithm~\ref{algo:check_pi-DC}) returns \texttt{NO} at line~\ref{algo:check_pi-DC:l15}.
When $\sigma'_\gamma$ is a viable and dynamic ES for $\Gamma'_\gamma$ at line~\ref{algo:check_pi-DC:l5},
then $\texttt{Check-\oDCC}()$ (Algorithm~\ref{algo:check_pi-DC})  proceeds as follows. At line~\ref{algo:check_pi-DC:l6},
some $\eta\in [0,1)$ is computed as in the proof of Lemma~\ref{lem:reduction_hard},
\ie such that $[\sigma'_\gamma(s)]_v-\eta-k \in [0, |V|\cdot \gamma) \text{ holds for \emph{no} } v\in V, s\in\Sigma_P, k\in \Z$.
Notice that it is easy to find such $\eta$ in practice. Indeed,
one may view the real semi-open interval $[0,1)$ as if it was partitioned
into chunks (\ie smaller semi-open intervals) of length $\gamma$;
as observed in the proof of Lemma~\ref{lem:reduction_hard}, there are only $|\Sigma_P|\cdot |V|$
choices of pairs $(s,v)\in\Sigma_P\times V$, and each pair rules out a (circular)
semi-open interval of length $|V|\cdot \gamma$; therefore,
there is at least one chunk of length $l_\gamma \geq \gamma$,
within $[0,1)$, where $\eta$ can be placed, and we can easily find it just by
inspecting (exhaustively) the pairs $(s,v)\in\Sigma_P\times V$.
In fact, the algorithm underlying Theorem~\ref{thm:mainresult} always deliver an \emph{earliest} ES
(\ie one in which the time values are the smallest possible, in the space of all consistent ES{s}),
so that for each interval of length $|V|\cdot\gamma$,
the only time values that we really need to check and rule out are $|V|$ multiples of $\gamma$.
Therefore, at line~\ref{algo:check_pi-DC:l6}, $\eta$ exists and it can be easily found in time $O(|\Sigma_P|\cdot |V|^2)$.
So, at line~\ref{algo:check_pi-DC:l7}, for each $s\in \Sigma_P$ and $v\in V^+_s$,
the value $[\sigma'_\gamma(s)]_v$ is shifted to the left by setting $[\sigma'_\gamma(s)]_v \leftarrow [\sigma'_\gamma(s)]_v - \eta$.
Then, the following $\pi$-ES $\sigma\in\S_\Gamma$ is constructed at lines~\ref{algo:check_pi-DC:l9}-\ref{algo:check_pi-DC:l13}:
for each $s\in\Sigma_P$ and $v\in V^+_s$, the execution-time is set
$[\sigma(s)]^t_v\leftarrow \big\lfloor [\sigma'_\gamma(s)]_v \big\rfloor$,
and the ordering $[\sigma(s)]^\pi$ follows the ordering on $P$ that is induced by $\sigma'_\gamma(s)$.
Finally, $(\texttt{YES}, \sigma)$ is returned to output at line~\ref{algo:check_pi-DC:l14}.

To conclude, we can prove the following.
\begin{proof}[Correctess and Complexity of Algorithm~\ref{algo:check_pi-DC}]
The correctness of Algorithm~\ref{algo:check_pi-DC} follows directly from Theorem~\ref{thm:reduction}~and~Theorem~\ref{thm:mainresult}~(Item~2),
plus the fact that $\eta\in[0,1)$ can be computed easily, at line~\ref{algo:check_pi-DC:l6},
as we have already mentioned above. The (pseudo) singly-exponential time complexity of
Algorithm~\ref{algo:check_pi-DC} follows from that of Theorem~\ref{thm:mainresult} (Item~2)
plus the fact that all the integer weights in $\Gamma$ are scaled-up by a factor $1/\gamma=|\Sigma_P|\cdot |V|^2+1$ in $\Gamma'_\gamma$;
also notice that $\eta\in[0,1)$ can be computed in time $O(|\Sigma_P|\cdot|V|^2)$, as we have already mentioned.
Therefore, all in, the time complexity stated in Theorem~\ref{thm:mainresult} (Item~2)
increases by a factor $1/\gamma=|\Sigma_P|\cdot |V|^2+1$; thus, it is
$O\big(|\Sigma_P|^{4}|A|^2|V|^3 +	|\Sigma_P|^5|A||V|^4|P| + |\Sigma_P|^6|V|^5|P|\big)W$.
In the next paragraph we shall improve this time complexity by a factor $|\Sigma_P|\cdot |V|$.
\end{proof}

\subsubsection{Improving Algorithm~\ref{algo:check_pi-DC}: Reducing From \oDCC to $1$-DC.}
Here below, we present a reduction from \oDCC to $1$-DC, which basically implies Theorem~\ref{thm:mainresult_pi-DC}.
The main idea is to scale-up all the weights of the input \CSTN $\Gamma$
by a (sufficiently large, integer) factor $\zeta$, then to relax all of the corresponding LTC{s} by an additive amount $|V|$,
and finally to observe that $1$-dynamicity can be inferred from $\pi$-dynamicity (see the proof of Lemma~\ref{lem:reduction_easy1dc}).
A similar argument as that of Lemma~\ref{lem:reduction_hard} establishes the converse direction (see the proof of Lemma~\ref{lem:reduction_hard1dc}).
The resulting (scaled-up) relaxed \CSTN will be denoted by $\Gamma^*_\zeta$ and it is defined next.

\begin{Def}{Relaxed \CSTN $\Gamma^*_\zeta$.}\label{def:relaxed_cstn1dc}
Let $\Gamma=( V, A, L, \Ord, \Ord{V}, P )$ be any \CSTN with integer constraints. Let $\zeta \in \N$.
Define $\Gamma^*_\zeta \triangleq ( V, A^*_\zeta, L, \Ord, \Ord{V}, P )$
to be a \CSTN that differs from $\Gamma$ only in the numbers appearing in the constraints.
Specifically, each constraint $( u-v \leq \delta, \ell )\in A$
is replaced in $\Gamma^*_\zeta$ by a scaled-up and slightly relaxed constraint,
$( u-v\leq \delta^*_\zeta, \ell ) \in A^*_\zeta$, where:
\[
	\delta^*_\zeta\triangleq \delta\cdot\zeta+|V|. 
\]
\end{Def}
The following two lemmata hold for any \CSTN $\Gamma$.

\begin{Lem}\label{lem:reduction_easy1dc}
Let $\zeta \geq |V|$ be any integer. If $\Gamma$ is \oDCC, then $\Gamma^*_\zeta$ is $1$-DC.
\end{Lem}
\begin{proof}
Since $\Gamma$ is \oDCC, by Proposition~\ref{prop:integral}, there exists an integer valued,
viable and $\pi$-dynamic $\pi$-ES $\sigma$ for $\Gamma$.

Let us fix some integer $\zeta\geq |V|$.
Define the ES $\sigma^*_\zeta\in\S_{\Gamma^*_\zeta}$ as follows, for every $s\in\Sigma_P$ and $v\in V^+_s$:
\[ [\sigma^*_\zeta(s)]_v \triangleq [\sigma(s)]^t_v\cdot \zeta +[\sigma(s)]^\pi_v.\]
We claim that the total ordering of the values $\{[\sigma^*_\zeta(s)]_v\}_{v\in V^+_s}$, for a given $s\in \Sigma_P$,
	coincides with that of $([\sigma(s)]^t, [\sigma(s)]^\pi)$. Indeed, for any $s\in \Sigma_P$ and $u,v\in V^+_s$,
\begin{align*}
	[\sigma^*_\zeta(s)]_u < [\sigma^*_\zeta(s)]_v & \iff
		[\sigma(s)]^t_u\cdot \zeta + [\sigma(s)]^\pi_u < [\sigma(s)]^t_v\cdot \zeta + [\sigma(s)]^\pi_v \\
	 & \iff [\sigma(s)]^t_u < [\sigma(s)]^t_v \text{ or } (\, [\sigma(s)]^t_u = [\sigma(s)]^t_v \text{ and } [\sigma(s)]^\pi_u < [\sigma(s)]^\pi_v \, ),
\end{align*}
because $1\leq [\sigma(s)]^\pi \leq |V|$ and $\zeta \geq |V|$, and they are all integers as well as $[\sigma(s)]^t$, so the claim holds.

Since they are all integers, for any $s\in \Sigma_P$ and $u,v\in V^+_s$, if $[\sigma^*_\zeta(s)]_u < [\sigma^*_\zeta(s)]_v$,
	then $[\sigma^*_\zeta(s)]_v - [\sigma^*_\zeta(s)]_u \geq 1$.

At this point, the fact that $\sigma^*_\zeta$ is $1$-dynamic follows directly from the $\pi$-dynamicity of $\sigma$.

Moreover, no LTC $(u-v\leq \delta^*_\zeta, \ell)$ of $\Gamma^*_\zeta$ is violated in any scenario $s\in\Sigma_{P}$, since,
if $\Delta^*_{\zeta, u,v} \triangleq [\sigma^*_\zeta(s)]_u - [\sigma^*_\zeta(s)]_v$,
\begin{align*}
\text{then: }\;\;\;\; \Delta^*_{\zeta, u,v} & = \big([\sigma(s)]^t_u\cdot\zeta + [\sigma(s)]^\pi_u \big)
								- \big([\sigma(s)]^t_v \cdot\zeta + [\sigma(s)]^\pi_v \big) \\
				& \leq ([\sigma(s)]^t_u - [\sigma(s)]^t_v)\cdot \zeta + |V| \\
				& \leq \delta\cdot \zeta + |V| = \delta^*_\zeta.
\end{align*}
So, $\sigma^*_\zeta$ is viable. Since $\sigma^*_\zeta$ is also $1$-dynamic, then $\Gamma^*_\zeta$ is $1$-DC.
\end{proof}

The next lemma shows that the converse direction holds as well, but for (exponentially) greater values of $\zeta$.
\begin{Lem}\label{lem:reduction_hard1dc} Let $\zeta > |\Sigma_P| \cdot |V|^2$ be any integer.
If $\Gamma^*_\zeta$ is $1$-DC, then $\Gamma$ is \oDCC. \end{Lem}
\begin{proof}
Let $\sigma^*_\zeta\in\S_{\Gamma^*_\zeta}$ be some viable and $1$-dynamic ES for $\Gamma^*_\zeta$.

Firstly, we aim at showing that, w.l.o.g., the following lower bound holds:
\[ \frac{[\sigma^*_\zeta(s)]_v}{\zeta} - \left\lfloor \frac{ [\sigma^*_\zeta(s)]_v}{\zeta} \right\rfloor \geq \frac{|V|}{\zeta},
\text{ for \emph{all} } s\in \Sigma_{P} \text{ and } v\in V^+_s. \hspace*{.75cm} \text{(LB*)} \]

This will allow us to simplify the rest of the proof.
In order to prove (LB*), let us pick any $\eta\in[0,1)$ such that:
\[ \frac{[\sigma^*_\zeta(s)]_v}{\zeta}-\eta-k \in [0, \frac{|V|}{\zeta} ),\; \text{ for \emph{no} } v\in V,\; s\in\Sigma_P,\; k\in \Z. \]
Observe that such a value $\eta$ exists. Indeed, there are only $|\Sigma_P|\cdot |V|$
choices of pairs $(s,v)\in\Sigma_P\times V$ and each pair rules out a (circular)
semi-open interval of length $\frac{|V|}{\zeta}$ in $[0,1)$, so the total measure of invalid values for $\eta$ in the
semi-open real interval $[0,1)$ is at most $|\Sigma_P|\cdot |V|\cdot \frac{|V|}{\zeta} < 1$ (for $\zeta > |\Sigma_P| \cdot |V|^2$).
So $\eta$ exists.

By subtracting $\eta$ to all time values $\{\frac{[\sigma^*_\zeta(s)]_v}{\zeta}\}_{v\in V, s\in\Sigma_P}$,
	we can assume w.l.o.g. that $\eta=0$ holds for the rest of this proof; and thus, that (LB*) holds.
Now, let us define:
 \[ [\sigma(s)]^t_v \triangleq \left\lfloor \frac{[\sigma^*_\zeta(s)]_v}{\zeta} \right\rfloor \;\; \forall\, s\in\Sigma_P, v\in V^*_s; \]
and let $[\sigma(s)]^\pi$ be the total ordering induced by $\sigma^*_\zeta(s)$.
Observe that $\sigma$ is a well-defined $\pi$-ES (\ie that $[\sigma(s)]^\pi$ is coherent \wrt $[\sigma(s)]^t$),
thanks to the fact that $\lfloor \cdot \rfloor$ is a monotone operator.
Since for every $s\in\Sigma_P$ the total ordering $[\sigma(s)]^\pi$ is
	the same as that of $\sigma^*_\zeta(s)$ and $\sigma^*_\zeta$ is $1$-dynamic, then $\sigma$ is $\pi$-dynamic.

It remains to prove that $\sigma$ is viable. For this, take any constraint $(u-v\leq \delta, \ell)\in A$ in $\Gamma$,
and suppose that: \[ [\sigma^*_\zeta(s)]_u-[\sigma^*_\zeta(s)]_v\leq \delta^*_\zeta
	= \delta\cdot\zeta + |V|. \hspace*{3ex} \text{(A*)} \]
If $[\sigma^*_\zeta(s)]_u-[\sigma^*_\zeta(s)]_v\leq \delta\cdot \zeta$,
and $\delta\in \Z$,	then clearly $[\sigma(s)]^t_u - [\sigma(s)]^t_v \leq \delta$.
So, the interesting case is when:
	\[ 0 < [\sigma^*_\zeta(s)]_u - [\sigma^*_\zeta(s)]_v-\delta\cdot\zeta\leq |V|.\]
For this, we observe that the following ($*$) holds by (LB*):
	\[ \left\lfloor \frac{[\sigma^*_\zeta(s)]_u}{\zeta} \right\rfloor \leq \frac{[\sigma^*_\zeta(s)]_u}{\zeta} - \frac{|V|}{\zeta}. \hspace*{2ex} \text{($*$)}\]
Also, it is clear that:
	\[ \left\lfloor \frac{[\sigma^*_\zeta(s)]_v}{\zeta} \right\rfloor > \frac{[\sigma^*_\zeta(s)]_v}{\zeta} - 1. \hspace*{8ex} \text{($**$)}\]

Then,
\begin{align*} \hspace*{-2ex}
[\sigma(s)]^t_u - [\sigma(s)]^t_v & =
				\left\lfloor \frac{[\sigma^*_\zeta(s)]_u}{\zeta} \right\rfloor - \left\lfloor \frac{[\sigma^*_\zeta(s)]_v}{\zeta} \right\rfloor  & \text{(by def. of $[\sigma(s)]^t_{x\in\{u,v\}}$)}\\
				  & < \left( \frac{[\sigma^*_\zeta(s)]_u}{\zeta}-\frac{|V|}{\zeta}\right) - \left( \frac{[\sigma^*_\zeta(s)]_v}{\zeta} - 1 \right) & \text{(by ($*$) and ($**$))}\\
					  & = \frac{1}{\zeta}\left([\sigma^*_\zeta(s)]_u - [\sigma^*_\zeta(s)]_v\right) - \frac{|V|}{\zeta} + 1 & \text{(by rewriting)} \\
						& \leq \frac{1}{\zeta}\left( \delta\cdot\zeta+|V|\right) - \frac{|V|}{\zeta} + 1 & \text{(by (A*))} \\
				  & = \delta + 1. &
\end{align*}
Now, since we have the strict inequality $[\sigma(s)]^t_u - [\sigma(s)]^t_v < \delta + 1$,
and since $[\sigma(s)]^t_u - [\sigma(s)]^t_v\in\Z$, then $[\sigma(s)]^t_u - [\sigma(s)]^t_v \leq \delta$, as desired.
So, $\sigma$ is viable. Since $\sigma$ is both viable and $\pi$-dynamic, then $\Gamma$ is \oDCC.
\end{proof}

In summary, Lemma~\ref{lem:reduction_easy1dc} and Lemma~\ref{lem:reduction_hard1dc} imply Theorem~\ref{thm:reduction1dc}.
\begin{Thm}\label{thm:reduction1dc}
Let $\Gamma$ be a \CSTN and let $\zeta > |\Sigma_P|\cdot|V|^2$ be any integer.
Then, $\Gamma$ is \oDCC if and only if $\Gamma^*_\zeta$ is $1$-DC.
\end{Thm}

Notice that, in Theorem~\ref{thm:reduction1dc},	it is fine to pick $\zeta \triangleq |\Sigma_P|\cdot|V|^2+1$.

Concerning the shift parameter $\eta$ in the proof of Lemma~\ref{lem:reduction_hard1dc},
it is not difficult to compute a correct $\eta$ in practice. Indeed, in this case,
one may view the real semi-open interval $[0,1)$ as if it was partitioned into chunks of length $1/\zeta$.
Since there are only $|\Sigma_P|\cdot |V|$ choices of pairs $(s,v)\in\Sigma_P\times V$, and each pair rules out a (circular)
semi-open interval of length $\frac{|V|}{\zeta}$, there is at least one chunk within $[0,1)$ of length,
	\[ l_\zeta \geq 1 - |\Sigma_P|\cdot |V|\cdot\frac{|V|}{\zeta} = \frac{1}{\zeta}, \]
where $\eta$ can be placed, and we can easily find it just by inspecting (exhaustively) the pairs $(s,v)\in\Sigma_P\times V$.
For each interval of length $\frac{|V|}{\zeta}$, the only time values that we really
	need to check and rule out are $|V|$ multiples of $\frac{1}{\zeta}$.

Therefore, $\eta$ can be found in time $O(|\Sigma_P|\cdot |V|^2)$.

In summary, in order to check \oDCC by reduction to $1$-\DCC, Algorithm~\ref{algo:check_pi-DC} can be modified as follows:

\paragraph{Checking \oDCC by Reduction to $1$-\DCC (Modifying Algorithm~\ref{algo:check_pi-DC})}
At line~\ref{algo:check_pi-DC:l1}, set $\zeta\leftarrow |\Sigma_P|\cdot |V|^2+1$.
Then, at lines~\ref{algo:check_pi-DC:l2}-\ref{algo:check_pi-DC:l3},
$\Gamma^*_\zeta$ is constructed as in Definition~\ref{def:relaxed_cstn1dc},
\ie $\Gamma^*_\zeta\leftarrow ( V, A^*_\zeta, L, \Ord, {\Ord}V, P )$,
where $A^*_\zeta\leftarrow \{( u-v\leq \delta\cdot\zeta+|V|, \ell ) \mid ( u-v\leq \delta, \ell )\in A \}$.
At this point, at line~\ref{algo:check_pi-DC:l5}, the
$1$-\DCC-Checking algorithm of Theorem~\ref{thm:mainresult} (Item~1) is invoked on input $( \Gamma^*_\zeta,1)$.
Let $\sigma^*_\zeta$ be its output. If $\Gamma^*_\zeta$ is not DC,
then $\texttt{Check-\oDCC}()$ (Algorithm~\ref{algo:check_pi-DC}) returns \texttt{NO} at line~\ref{algo:check_pi-DC:l15}.
Instead, when $\sigma^*_\zeta$ is a viable and dynamic ES for $\Gamma^*_\zeta$ at line~\ref{algo:check_pi-DC:l5},
then $\texttt{Check-\oDCC}()$ (Algorithm~\ref{algo:check_pi-DC}) proceeds as follows. At line~\ref{algo:check_pi-DC:l6},
some $\eta\in [0,1)$ is computed as in the proof of Lemma~\ref{lem:reduction_hard1dc},
\ie such that $[\sigma^*_\zeta(s)]_v/\zeta-\eta-k \in [0, |V|/\zeta ) \text{ holds for \emph{no} } v\in V, s\in\Sigma_P, k\in \Z$.
So, at line~\ref{algo:check_pi-DC:l7}, for each $s\in \Sigma_P$ and $v\in V^+_s$,
the value $[\sigma^*_\zeta(s)]_v/\zeta$ is shifted by $-\eta$.
Then, the following $\pi$-ES $\sigma\in\S_\Gamma$ is constructed at lines~\ref{algo:check_pi-DC:l9}-\ref{algo:check_pi-DC:l13}:
for each $s\in\Sigma_P$ and $v\in V^+_s$, the execution-time is set
$ [\sigma(s)]^t_v \triangleq \big\lfloor [\sigma^*_\zeta(s)]_v/\zeta \big\rfloor$,
and the ordering $[\sigma(s)]^\pi$ follows the ordering on $P$ induced by $\sigma^*_\zeta(s)$.
Finally, $(\texttt{YES}, \sigma)$ is returned to output at line~\ref{algo:check_pi-DC:l14}.

\begin{proof}[Proof of Theorem~\ref{thm:mainresult_pi-DC}]
	The correctness of the algorithm above follows directly from Theorem~\ref{thm:reduction1dc}~and~Theorem~\ref{thm:mainresult}~(Item~1),
	The (pseudo) singly-exponential time complexity follows from that of Theorem~\ref{thm:mainresult} (Item~1)
	plus the fact that all the integer weights in $\Gamma$ are scaled-up by a factor $\zeta=|\Sigma_P|\cdot |V|^2+1$ in $\Gamma^*_\zeta$;
	also recall that $\eta\in[0,1)$ can be computed in time $O(|\Sigma_P|\cdot|V|^2)$, as we have already mentioned.
	In summary, the time complexity stated in Theorem~\ref{thm:mainresult} (Item~1)
	increases by a factor $|\Sigma_P|\cdot |V|^2$.
	Note this improves the time complexity of Algorithm~\ref{algo:check_pi-DC} by a factor $|\Sigma_P|\cdot |V|$.
\end{proof}

\section{Conclusion}\label{sect:conclusion}
The notion of $\varepsilon$-DC has been introduced and analysed in~\citet{CR15} where an algorithm was also given to check
whether a \CSTN is $\varepsilon$-DC. By the interplay between $\varepsilon$-DC and the standard notion of DC,
also disclosed in~\citet{CR15}, this delivered the first (pseudo) singly-exponential time algorithm
checking whether a \CSTN is DC (essentially, DC-Checking reduces to $\varepsilon$-DC-Checking for a suitable value of $\varepsilon$).
In this paper, we proposed and formally defined \oDCC, a natural and sound notion of DC for \CSTN{s} in which the planner is
allowed to react instantaneously to the observations that are made during the execution.
A neat counter-example shows that \oDCC with instantaneous reaction-time is not just the special
case of $\varepsilon$-DC with $\varepsilon = 0$. Therefore, in this work, we offered the first sound-and-complete \oDCC-Checking algorithm for \CSTN{s}.
The time complexity of the procedure is still (pseudo) singly-exponential in the number $|P|$ of propositional letters.
The solution is based on a simple reduction from \oDCC-Checking to DC-Checking of \CSTN{s};
finally, we point out a reduction from \oDCC to $1$-\DCC, this allows us to further improve
the time complexity of the proposed algorithm by a factor $|\Sigma_P|\cdot |V|$.

\paragraph{Acknowledgments}
This work was supported by the \emph{Department of Computer Science, University of Verona, Italy}.

\section*{References}
\bibliographystyle{elsarticle-num-names-alpha}
\bibliography{biblio}
\end{document}